\documentclass[a4paper,12pt]{article}
\usepackage[top=2.5cm,bottom=2.5cm,left=2.5cm,right=2.5cm]{geometry}
\usepackage[margin=1cm,%
font=small,%
format=hang,%
labelsep=period,%
labelfont=bf]{caption}
\usepackage[utf8]{inputenc}
\usepackage{amsmath}
\usepackage{amsfonts}
\usepackage{amssymb}
\usepackage{cite}
\usepackage{amsthm}
\usepackage{makeidx}
\usepackage{graphicx}
\usepackage{mathrsfs,xcolor}
\usepackage{blkarray}
\usepackage{bm}
\usepackage{setspace}
\usepackage[left,pagewise,displaymath, mathlines]{lineno}
\usepackage{enumitem}
\usepackage{multirow}
\usepackage{hhline}
\usepackage{adjustbox}

\makeatletter
\newcommand\xleftrightarrow[2][]{\ext@arrow 0099{\longleftrightarrowfill@}{#1}{#2}}
\def\longleftrightarrowfill@{\arrowfill@\leftarrow\relbar\rightarrow}
\makeatother

\numberwithin{equation}{section}

\theoremstyle{plain}
\newtheorem{theorem}{Theorem}[section]

\newtheorem{proposition}{Proposition}[section]
\newtheorem{corollary}{Corollary}[section]

\theoremstyle{definition}
\newtheorem{definition}{Definition}[section]
\newtheorem{example}{Example}
\newtheorem{remark}{Remark}[section]

\usepackage{authblk}

\author[1]{ \textbf{Noel Fortun}}\author[2]{\textbf{Piolo Gaspar}}\author[2,*]{\textbf{Editha Jose}}\author[1]{ \textbf{Angelyn Lao}}\author[1,4,5]{\textbf{Eduardo Mendoza}}\author[3]{ \textbf{Luis Razon}}

\affil[1]{\footnotesize \textit{Department of Mathematics and Statistics, De La Salle University, Manila, 0922, Philippines}}
\affil[2]{\footnotesize \textit{Institute of Mathematical Sciences, University of the Philippines Los Ba\~nos, Los Ba\~nos, 4031, Laguna, Philippines}}
\affil[3]{\footnotesize \textit{Department of Chemical Engineering, De La Salle University, Manila, 0922, Philippines}}
\affil[4]{\textit{Center for Natural Sciences and Environmental Research, De La Salle University, Manila, 0922, Philippines}}
\affil[5]{\textit{Max Planck Institute of Biochemistry, Martinsried near Munich, Germany}}
\affil[*]{Corresponding author: \texttt{ecjose1@up.edu.ph}}

\title{\textbf{A reaction network approach to modeling carbon dioxide removal systems}}

\date{}

\begin{document}
\maketitle
\thispagestyle{empty}

\abstract{This paper focuses on what we call Reaction Network Cardon Dioxide Removal (RNDCR) framework to analyze several proposed negative emissions technologies (NETs) so as to determine when present-day Earth carbon cycle system would exhibit multistationarity (steady-state multiplicity) or possibly monostationarity, that will in effect lower the rising earth temperature. Using mathematical modeling based on the techniques of chemical reaction network theory (CRNT), we propose an RNDCR system consisting of the Anderies subnetwork, fossil fuel emission reaction, the carbon capture subnetwork and the carbon storage subnetwork. The RNCDR framework analysis was done in the cases of two NETs: Bioenergy with Carbon Capture and Storage and Afforestation/Reforestation. It was found out that these two methods of carbon dioxide removal are almost similar with respect to their network properties and their capacities to exhibit multistationarity and absolute concentration robustness in certain species.} 






\section{Introduction}\label{intro}

The rapidly increasing concentration of CO2 in the atmosphere and its dire consequences are one of the major pressing concerns worldwide at present. The Intergovernmental Panel on Climate Change (IPCC) has stated that more direct interventions like carbon dioxide removal (CDR) are now necessary to achieve net-zero greenhouse gas emissions (\cite{SHUKLA2022}). Predicting the impact of human interventions on the Earth system using mathematical models is thus a key task that would allow the selection and optimization of these activities prior to actual implementation (\cite{TAN2022}). The primary approach in Earth system modeling is to subdivide the entire Earth surface into three-dimensional, model the processes in each cell, and then subsequently model the inter-cell interactions. Ultimately, this may be the most accurate and reliable approach, but it requires large amounts of computing time even on the most powerful computers in the world (\cite{CB2018}).

Another approach is to model each major part of the Earth into separate compartments such as the atmosphere, ocean, and land without any spatial resolution and model them as interacting perfectly mixed chemical reactors. Such an approach has been taken in the following papers:  \cite{LENTON2000},  \cite{SCHMITZ2002},  \cite{ANDE2013}, and \cite{HEITZIG2016}. This approach has been further extended to include even economic systems ( \cite{ABDFHR2023}, \cite{NITZBON2017}). The computations involved are simpler and may provide useful insights that may be difficult to discern in more complex models. 

Numerical analysis of such systems has shown that the Earth System may exhibit so-called “tipping points”. The term tipping point captures the thought that  “… at a particular moment in time, a small change can have large, long-term consequences for a system” (\cite{MCKAY2022}). Exceeding a critical threshold in a key parameter may cause a rapid transition to another (perhaps undesirable) from which recovery may be difficult if not impossible. Tipping points would be commonly associated with systems that exhibit steady-state multiplicity (often called, interchangeably, multistationarity or multiple stationary states.) In the analysis of chemically reacting systems, these are more often called bifurcation points (\cite{RAZON1987}).
The numerical analysis performed in \cite{ANDE2013} demonstrated that pre-industrial Earth may be expected to exhibit only unique stable steady states. However, their analysis of industrial Earth showed that there may be regions in the parameter space where steady-state multiplicity may exist and that it is possible that a tipping point may be reached in the future.

Is it possible then that there may be human interventions that could restore the pre-industrial Earth with unique stable steady states?  The question could be also put this way: Under what conditions would a human intervention on Earth result in a system with unique steady states under all parameter combinations? Is there a way we can quickly screen possible interventions to find out?

The set of techniques called Chemical Reaction Network Theory (CRNT) is a method by which observations on the network properties of lumped chemically reacting systems can determine whether the given system can admit the possible existence (or absence) of steady-state multiplicity (\cite{FEIN2019}). Our previous studies (\cite{FMRL2018,FMRL2019,FOME2023}) that apply CRNT have shown that the pre-industrial Earth model proposed in \cite{SCHMITZ2002} has a unique steady state, while the capacity for steady-state multiplicity of the model in \cite{ANDE2013} depends on some parameters. In a subsequent study, we determined necessary conditions for the existence and uniqueness of steady states in a model of an industrial Earth system wherein Direct Air Capture is implemented (\cite{FLMR2024}).

In this paper, we demonstrate a framework for analyzing proposed negative emissions technologies (NET) and determining the conditions as to when present-day Earth can would still be exhibiting steady-state multiplicity. Consequently, this analysis would also show when present-day Earth would have unique steady states, making recovery from a high-temperature state more likely to succeed. 

\section{Power-law kinetic representation of an Earth system model}

The framework for analysis of the global carbon cycle models described in Section \ref{RNCDR} requires constructing a \textbf{power-law kinetic representation}. A power-law kinetic representation of a given dynamical system refers to a chemical reaction network (CRN) with power-law kinetics that is dynamically equivalent to the given system; i.e., they have identical ODEs.  The aim is to understand the dynamics of the power-law system by examining its power-law kinetic representation in a way that bypasses the numerical computations and simulations typically associated with nonlinear ODEs. Utilizing established tools and findings from CRNT, the system is analyzed with minimal dependence on particular parameters, as this approach does not rely on rate constants and handles kinetic orders symbolically. 

To clarify the concepts discussed in this section, we recall and use as an example the analysis conducted in \cite{FMRL2018,FOME2023,MJS2018} on the power-law representation of the pre-industrial carbon cycle of Anderies et al. \cite{ANDE2013}. The dynamics of the system is described by the following ODE system:
\begin{equation}\label{eq:Anderies}
\left.
 \begin{array}{cl}
\dot{A_1}&=k_1A_1^{p_1}A_2^{q_1} - k_2A_1^{p_2}A_2^{q_2}  \\
\dot{A_2 }&= k_2A_1^{p_2}A_2^{q_2} - k_1A_1^{p_1}A_2^{q_1} - a_mA_2 +a_m\beta A_3  \\
\dot{A_3} &=  a_mA_2 - a_m\beta A_3 \\
 \end{array}
 \right \},
\end{equation}
where $A_1$, $A_2$, and $A_3$ denote the carbon quantities in land biota, the atmosphere, and the ocean, respectively. The schematic representation of carbon transfer is recalled in Figure \ref{fig:anderies}.

\begin{figure}[t]%
\centering
\includegraphics[width=0.7\textwidth]{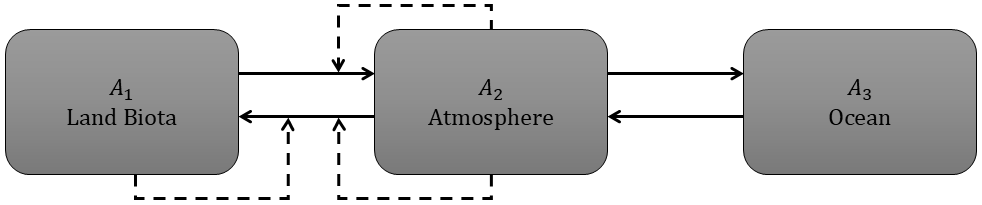}
\caption{The pre-industrial carbon cycle model based from Anderies et al. \cite{ANDE2013}. In the box model, the boxes represent the different pools, solid arrows indicate the transfer of carbon from one pool to another, and dashed arrows indicate the pools that influence a carbon transfer.}\label{fig:anderies}
\end{figure}

\subsection{Chemical reaction network representation}\label{section: CRN}
A chemical reaction network or CRN  is a finite collection of interdependent reactions occurring simultaneously. In a broader context, it can represent any system whose evolution relies on transforming its elements into different elements. The fundamental element of a chemical reaction is known as a \textbf{species}. Chemical species can represent a variety of entities, such as chemical elements, molecules, or proteins. In this study, the species symbolize the different carbon pools present in the system. A \textbf{complex} is defined as a nonnegative linear combination of these species. In other words, a complex is a set of species with associated nonnegative coefficients referred to as \textbf{stoichiometric coefficients}. Typically, a chemical reaction is expressed as:
$$
\text{Reactant complex} \rightarrow \text{Product complex},
$$
where the set of species on the left side of the equation (the \textbf{reactant complex}) are consumed or transformed to form the set of species on the right side (the \textbf{product complex}). Each complex in a CRN can be represented as a vector in a space known as \textbf{species space}, where the coordinates indicate the coefficients or stoichiometry of the various species. Consequently, each reaction can be associated to a vector called the reaction vector.  A \textbf{reaction vector} is formed by subtracting the reactant complex vector from the product complex vector.

Given a box model of a global carbon cycle, a corresponding \textbf{CRN representation} can be set up using the procedure proposed by Arceo et al. \cite{AJMSM2015}. In this approach, one associates the reaction $A_i \rightarrow A_j$ to the carbon transfer from pool $A_i$ to pool $A_j$. Moreover, if the carbon transfer is influenced by some carbon pools (as indicated by the dashed arrows in the schematic diagram), say $\sum A_k$, all these species are added to both sides of $A_i \rightarrow A_j$ to form the chemical reaction 
\begin{equation}\label{eq:CRNrep}
\underbrace{A_i+ \left( \sum A_k \right)}_{\mbox{reactant complex}} \rightarrow \underbrace{A_j + \left( \sum A_k \right) }_{\mbox{product complex}}
\end{equation}
This process preserves the coordinates of the reaction vectors, which is important in describing the dynamics of the whole system (see Section \ref{sec:PL}).

In the box model of the pre-industrial model shown in Figure \ref{fig:anderies}, the transfer of carbon from the atmosphere $A_2$ to land $A_1$ is influenced by $A_1$ and $A_2$. Thus, the reaction associated with this process is  $A_2+ \{ A_1+A_2 \} \rightarrow A_1+\{A_1+A_2\}$ or simply $A_1+2A_2 \rightarrow 2A_1+A_2$. The carbon transfer from land to atmosphere is represented by the reaction $A_1 +A_2 \to 2A_2$ because the process is influenced by $A_2$. The movement of carbon between the ocean and the atmosphere is assumed to be unaffected by the carbon pools. Hence, the reversible reaction $A_2 \rightleftarrows A_3$ is obtained. In summary, a CRN representation of the Anderies pre-industrial system  is the following network consisting of three species -- $A_1$, $A_2$, and $A_3$ -- and four reactions (each denoted by $R_i$):

\begin{equation}\label{eq: CRN}
\begin{blockarray}{lcccc}
\text{Reaction} & \text{Reactant} & \text{Product} & \text{Reaction vector} & \text{Reaction rate}  \\
R_1: A_1+2A_2 \rightarrow 2A_1+ A_2 & [1, 2, 0]^\top & [2, 1, 0]^\top & [1, -1, 0]^\top  & k_1A_1^{p_1}A_2^{q_1}\\
R_2: A_1+A_2 \rightarrow 2A_2 & [1, 1, 0]^\top & [0, 2, 0]^\top & [-1, 1, 0]^\top & k_2A_1^{p_2}A_2^{q_2} \\
R_3: A_2 \rightarrow A_3 & [0, 1, 0]^\top & [0, 0, 1]^\top & [0, -1, 1]^\top & a_m A_2\\
R_4: A_3 \rightarrow A_2 & [0, 0, 1]^\top & [0, 1, 0]^\top & [0, 1, -1]^\top & a_m\beta A_3\\
\end{blockarray}
\end{equation}

When considered as a directed graph, a CRN is classified as \textbf{weakly reversible} if the presence of a path from one complex $C_i$ to another complex $C_j$ guarantees a path from $C_j$ back to $C_i$. Clearly, the CRN in (\ref{eq: CRN}) is not weakly reversible. 
A collection of complexes interconnected by arrows is known as a \textbf{linkage class}. In the above CRN, there are three linkage classes: $\{A_1 + 2A_2 \rightarrow  2A_1 + A_2\}$, $ \{ A_1+A_2 \rightarrow 2A_2 \}$ and $\{A_2 \rightleftarrows A_3\}$.

The span or the set of all possible linear combinations of the reaction vectors is called the \textbf{stoichiometric subspace} of the network. The \textbf{rank} of a CRN refers to the dimension of the stoichiometric subspace (i.e., the maximum number of linearly independent reaction vectors). 
In (\ref{eq: CRN}), its  stoichiometric subspace has 2 basis vectors namely $\{ [1, -1, 0]^\top, [0, 1, -1]^\top  \}$. Hence, the rank of the CRN is 2.

Many significant outcomes in CRNT focus on the nonnegative structural index known as \textbf{deficiency}. The deficiency of a CRN is the non-negative integer determined by subtracting the number of linkage classes $\ell$ and the rank $s$ from the total number of complexes $n$. This index remains unaffected by the size of the network. Even large or complex CRNs can exhibit a deficiency of zero. The deficiency reflects the degree of `linear independence' in reactions; a greater deficiency indicates a lower level of linear independence (\cite{SHFE2012}).

The CRN of the pre-industrial carbon cycle model shown (\ref{eq: CRN}) has deficiency $\delta=1$ since it has $n=5$ complexes, $\ell=2$ linkage classes, and rank $s=2$. In \cite{MJS2018, FOME2021}, a dynamically equivalent deficiency-zero CRN representation of this system was introduced by reaction translation (as described in \cite{JOHNSTON2014}) but without changing the reaction vector. Specifically, the second reaction $A_1 +A_2 \to 2A_2$ is translated by by adding $A_1$ to both sides of the reaction. Consequently, the CRN becomes a reversible (and hence, weakly reversible) deficiency-zero network:

\begin{equation}\label{eq:Anderies2}
\left.
 \begin{array}{rl}
A_1+2A_2 & \rightleftarrows 2A_1+A_2 \\
A_2 & \rightleftarrows A_3.
 \end{array}
 \right.
\end{equation}

\subsection{Power-law kinetics}\label{sec:PL}

The dynamics of a network are governed by a set of nonlinear ordinary differential equations (ODEs), derived from the CRN along with the definition of the reaction rate functions, or kinetics. This set of ODEs can be expressed in vector notation, summarizing a group of equations, where each equation corresponds to the evolution of a specific species in the CRN. Formally, the system may be expressed as: 
\begin{equation}\label{eq:ode}
\dot{x}= \sum_\mathscr{R} \kappa_{y \rightarrow y'} (x) (y'-y)    
\end{equation}
where $x$ denotes the vector of species composition. The overdot indicates differentiation with respect to time. The scalar $\kappa_{y \rightarrow y'} (x)$ encodes the rate at which the reaction $y \rightarrow y'$ occurs. The quantity $(y'-y)$ pertains to the reaction vector associated with the reaction $y \rightarrow y'$. The symbol $\mathscr{R}$ denotes the set of all reactions in the given CRN and its presence under the summation sign tells that the sum is taken over the reactions in the CRN. A \textbf{positive steady state} of the system is a positive species composition $x$ for which $\dot{x}=0$. 

Equation (\ref{eq:ode}) highlights the significance of the stoichiometric subspace in a CRN, as it sets limits on the system's dynamics. Although the species compositions change over time, their trajectories cannot freely roam throughout the species space. That is, the concentrations of the species are restricted to the translations of the stoichiometric subspace, known as \textbf{stoichiometric compatibility classes}. A system is \textbf{multistationary} or has the capability for multiple steady states if there is at least one stoichiometric compatibility class with at most two distinct positive steady states. Otherwise, the system is \textbf{monostationary}.

The CRN representation of a global carbon cycle must be endowed with \textbf{power-law kinetics} in order to reflect the dynamics described in Section \ref{RNCDR}; that is, the functions that govern all the reactions are power-law functions. The power law functions of a CRN representation are encoded using the \textbf{kinetic order matrix} $F$, where entry $F_{ij}$ encodes the kinetic order of the $j$th species in the $i$th reaction. Thus, for power law kinetic systems, the rate functions described in the ODE system in (\ref{eq:ode}) can be specified as
$$\displaystyle K_{i}(x)=\kappa_i x^{F_{i,*}} \quad \text{for all } i \in \mathscr{R},$$
where $F_{i,*}$ is the row vector containing the kinetic orders of the species of the reactant complex in the $i$th reaction. Hence, given the power-law rates in (\ref{eq: CRN}), the kinetic order matrix of the current system is 
\begin{equation*}
F= \begin{blockarray}{cccl}
A_1 & A_2 & A_3  & \\
\begin{block}{[ccc]l}
p_1 & q_1 & 0 & R_1 \\
p_2 & q_2 & 0 &  R_2 \\
0 & 1 & 0 &  R_3 \\
0 & 0 & 1 &  R_4 \\
\end{block}
\end{blockarray}.
\end{equation*}

Let $N$ be the matrix, called \textbf{stoichiometric matrix}, whose columns are the reaction vectors of the CRN. Then the system of ODEs specified in Equation (\ref{eq:ode}) can be written as
$$
\dot{x}=NK(x),
$$
where $K(x)$ is the vector which contains the reaction rates. To illustrate this, we note that the ODE system of the kinetic representation of the Anderies pre-industrial system is
$$
\begin{bmatrix} \dot{A_1} \\ \dot{A_2} \\ \dot{A_3} \end{bmatrix}= \begin{bmatrix} 1 & -1 & 0 & 0 \\ -1 & 1 & -1 & 1 \\ 0 & 0 & 1 & -1 \end{bmatrix}\begin{bmatrix} k_1A_1^{p_1}A_2^{q_1} \\ k_2A_1^{p_2}A_2^{q_2} \\ a_m A_2 \\ a_m\beta A_3
\end{bmatrix},
$$
which is clearly similar to the ODE system in (\ref{eq:Anderies}) indicating the dynamical equivalence of the kinetic representation and the original system.
\\
\begin{remark}
Although the preliminary concepts presented here are fundamental, they are not exhaustive, particularly as preparation for the mathematical proofs outlined in the paper. For a formal presentation of the concepts, notations, and results in CRNT needed for the proof of the results in this study, refer to Appendix \ref{appendix: CRNT}.
\end{remark}

\section{The Reaction Network Carbon Dioxide Removal (RNCDR) Framework} \label{RNCDR}

In this section, we describe the structure of a Reaction Network Carbon Dioxide Removal (RNCDR) kinetic system. We highlight two components common to all models, the Anderies pre-industrial subsystem and the fossil fuel emission reaction. We then elaborate the remaining Carbon Dioxide Removal (CDR)-specific network and kinetics elements. We introduce the concept of CDR classes and illustrate how the Anderies decomposition of an RNCDR network may impact class characteristics using the Direct Air Capture (DAC) system as an example.

\subsection{A brief review of Anderies pre-industrial systems}\label{preindustrial}

The definition and some basic properties of an Anderies pre-industrial kinetic system are provided in the previous section. The concept of an Anderies pre-industrial system was derived in several steps from the model of Anderies et al. \cite{ANDE2013}, which serves as the formal basis of the influential Science publication by Steffen et al. \cite{STEF2015}. In \cite{FMRL2018}, a power-law approximation was constructed to model the dynamics of the system. The dynamically equivalent Anderies pre-industrial system, whose CRN representation is given in Eq. (\ref{eq:Anderies2}),  was introduced in \cite{FOME2021} and analyzed in detail in \cite{FOME2023}, where the Anderies system class concept was also introduced.

For $q_2-q_1\neq 0$, the ratio $R = \dfrac{p_2- p_1}{q_2-q_1}$ defined three classes: \textbf{positive}, \textbf{negative}, or \textbf{null} based on whether $R$ is positive, negative, or equal 0. For $p_2-p_1=0$, we now introduce a further class:
\\
\begin{definition}
A $\bm{Q}$\textbf{-null} Anderies preindustrial system is one with $p_2-p_1\neq 0$ and $q_2-q_1=0,$, that is, $Q = \dfrac{q_2- q_1}{p_2-p_1}$ is defined but equal to zero. We denote the class with $AND_{Q_0}.$    
\end{definition}
$ $
\begin{remark}
Henceforth, we omit the adjective “pre-industrial” and simply say “Anderies system”.  For clarity, we refer to null systems as P-null systems and denote the class with $AND_{P_0}.$ Table 1 provides additional notation and terminology for the Anderies systems.    
\end{remark}
$ $
\\
\indent In \cite{FOME2023}, it was shown that an Anderies $P$-null system is monostationary and demonstrates absolute concentration robustness in the species $A_1$. \textbf{Absolute concentration robustness} (ACR) is a condition where the concentration of a species within a network attains the same value across all positive steady states, regardless of the initial conditions (\cite{SHFE2010}). The following propositions show that $Q$-null systems have properties analogous to $P$-null systems:
\\
\begin{proposition}
Any $AND_{Q_0}$ system is monostationary and has ACR only in $A_1$.
\end{proposition}

\begin{proof}
By setting the ODE system to 0, we obtain the following parametrization of equilibria: $A_1 = \left(\dfrac{k_1}{k_2}\right)^{\frac{1}{p_2-p_1}}$, $A_2$ a free parameter, and $A_3 = \left(\dfrac{1}{\beta}\right)A_2.$  If $A_0$ denotes the total amount, which in this case is simply the sum of the coordinates since $S^\perp = \text{span } \{ [1,1,1]^\top \}$, then $A_2 = \left( \dfrac{\beta}{\beta +1} \right) (A_0 -A_1).$ If another positive equilibrium is in the same stoichiometric class, then $A_2'= A_2,$ implying coincidence of the equilibria. Hence, any $AND_{Q_0}$ system is monostationary.

On the other hand, since any Anderies system
is a PLP system, the species hyperplane criterion for ACR \cite{LLMM2022} guarantees the ACR in $A_1$ (see Appendix \ref{appendix:PLP}). Since $\widetilde{S}^\perp = \text{span } \{ [-Q, 1, 1]^\top \}$ and $Q = 0,$ it follows from the criterion that $A_1$ is an ACR species while $A_2,A_3$ are not. 
\end{proof}

As we will see in the following definitions, any CDR system in the RNCDR framework contains an Anderies subsystem. The following table collects additional notation and terminology related to Anderies systems:

\begin{table}[h]
\centering
\begin{tabular}{|c|c|}
\hline
 \textbf{Symbol} & \textbf{Name}  \\
\hline
$p_1, p_2$ & Kinetic orders of land biota interaction \\
\hline
$p_1$& Land biota photosynthesis interaction ($p$-interaction)  \\
\hline
$p_2$& Land biota respiration interaction ($r$-interaction) \\
\hline
$q_1, q_2$ & Kinetic orders of atmosphere interaction  \\
\hline
$q_1$ & Atmosphere photosynthesis interaction  ($p$-interaction) \\
\hline
$q_2$ & Atmosphere respiration interaction  ($r$-interaction) \\
\hline
$p_2-p_1$ & Land biota $r$-$p$-interaction difference\\
\hline
$q_2-q_1$ & Atmosphere $r$-$p$-interaction difference \\
\hline
$R=\dfrac{p_2-p_1}{q_2-q_1}, q_2-q_1 \neq 0$ & Land biota-atmosphere $r$-$p$-interaction difference ratio \\
\hline
$Q=\dfrac{q_2-q_1}{p_2-p_1}, p_2-p_1 \neq 0$ & Atmosphere-land  biota $r$-$p$-interaction difference ratio\\
\hline
\end{tabular}
\caption{Notations Related to Anderies Systems}
\end{table} 

\subsection{Species set of an RNCDR network}\label{species}

In addition to Anderies species $A_1, A_2$ and $A_3$, the species set of an RNCDR system contains $A_4$ denoted as the Total Carbon Stock  (TCS) and $A_i$, the CDR carbon storage for each member of the CDR portfolio.

\subsubsection{Total Carbon Stock}\label{total}

The total carbon is an extension of the  geological stock concept to account for two properties of CDR technologies. Some CDR systems such as those based on pyrolysis of biomass output both long-term storable carbon (biocharcoal) and short-term storable carbon (biofuels) with only the former accounting for fossil fuel emission. Furthermore, there are CDR processes such as weathering which output only inorganic carbon and does not contribute to fossil fuel emission too.

The total carbon stock hence contains the carbon subpools: long-term organic (contains the geological stock), long-term inorganic and short-term (contains both organic and inorganic). Only the long-term organic is considered as input for the fossil fuel emission kinetics (see Section \ref{kinetics}).

\subsubsection{CDR Storage Species}\label{CDR}

Next, we fix the species notation for the CDR-specific storage as shown in Table \ref{table:storage}. This convention facilitates not only the assignment of parameters to each  CDR method (see Table \ref{table:parameter}), it also minimizes changes to reaction network notation when combinations of CDR methods, i.e., portfolios are modeled.
\\
\begin{remark}
\begin{enumerate}
\item [1.] The following assignments have been made for potential extension of the pre-industrial model:
\begin{itemize}
    \item[-] $A_5$ for rocks or weathering pool
    \item[-] $A_6$ for deep ocean (ocean floor) system
    \item[-] $A_7$ for wetlands
\end{itemize}

\item [2.] To avoid long vectors with many zeros in between, we write e.g., $[A_1,A_2,A_3,A_4,A_9,A_{17}]$ for carbon composition in a system with DAC and DOC. 
 \end{enumerate}   
\end{remark}

\begin{table}[h]
\centering
\begin{tabular}{|c|c|}
\hline
 \textbf{Symbol} & \textbf{Name}  \\
\hline
$A_8$ &  Bioenergy with Carbon Capture and Storage (BECCS) Storage \\
\hline
$A_9$ & Direct air capture (DAC) Storage\\
\hline
$A_{10}$ & Enhanced Weathering (EW) Storage\\
\hline
$A_{11}$ & Biochar Storage\\
\hline
$A_{12}$ & Ocean Fertilization (OF) Storage\\
\hline
$A_{13}$ & Soil carbon sequestration (SCS) Storage\\
\hline
$A_{14}$ & Wetland Restoration (WR) Storage\\
\hline
$A_{15}$ & Afforestation/reforestation (AR) Storage\\
\hline
$A_{16}$ & Ocean alkalinization (OA) Storage\\
\hline
$A_{17}$ & Direct ocean capture (DOC) Storage\\
\hline
\end{tabular}
\caption{Notations for CDR-specific storage}
\label{table:storage}
\end{table}

\subsection{Complexes and reactions of an RNCDR}\label{complexes}

The common complexes and reactions of RNCDR systems consist of those of the Anderies subnetwork as well as $A_4$ and the fossil fuel emission reaction $A_4\rightarrow A_2.$ 
The system's carbon capture subnetwork will contain at least one reaction since there is an outgoing reaction from the carbon source for removal and an incoming  $\rightarrow A_i$ in it (which may coincide as in DAC). The carbon storage subnetwork will contain one reaction $A_i \rightarrow A_4$ for each CDR method in the portfolio. Hence, each RNCDR system will have at least 5 species, 6 complexes and 7 reactions.

\subsection{Kinetics of an RNCDR system}\label{kinetics}

For the kinetics discussion, it is useful to introduce the following standard numbering of the common reactions:
\begin{align*}
R_1: & A_1 + 2A_2 \rightarrow 2A_1 + A_2\\
R_2: & A_2 + 2A_1 \rightarrow A_1 + 2A_2\\
R_3: & A_2 \rightarrow A_3\\
R_4: & A_3 \rightarrow A_2\\
R_5: & A_4 \rightarrow A_2   
\end{align*}

We also denote the CDR-specific storage reactions with $R_{i,6}: A_i \rightarrow A_4,$ for $i= 8,\dots,17$

\subsubsection{TCS parameters of a CDR technology}\label{parameters}

We now introduce two parameters for each CDR method to specify the storage type of their captured carbon. For each CDR species $A_i$, the parameter $\lambda_i$ denotes the percentage of long-term carbon in its output. On the other hand, the parameter $\mu_i$ indicates the percentage of inorganic carbon in the output.
\\
\begin{remark}
For CDR technologies that store their fully organic output in geological stock such as BECCS and DAC, we have  $\lambda_i = 1$ and  $\mu_i=0.$
\end{remark}
$ $
\\
Table \ref{table:parameter} below collects the parameter values used in the current RNCDR framework.

\begin{table}[h]
\centering
\begin{tabular}{|c|c|c|c|}
\hline
 NET & Physical Storage & $\lambda_i$ & $\mu_i$ \\
\hline
BECCS $(A_8)$ & CO$_2$ injected into the geological stock & 1 & 0\\ 
\hline
DAC $(A_9)$ & CO$_2$ injected onto the geological stock & 1 & 0\\ 
\hline
EW $(A_{10})$ & Rock spread primarily on beaches (other fields) & 1 & 1\\ 
\hline
Biochar $(A_{11})$ & Biocharcoal in soil (0.5), biofuel in geostock (0.5) & 1 & 0.5\\ 
\hline
OF $(A_{12})$ & Deep ocean & 1 & 0\\ 
\hline
SCS $(A_{13})$ & Soil (microbial enhancement) & 0.01 & 0\\ 
\hline
WR $(A_{14})$ & Wetland floor & 0.01 & 0.9\\ 
\hline
AR $(A_{15})$ & Sequence of trees and soil & 0.5 & 0\\ 
\hline
OA $(A_{16})$ & Deep ocean & 1 & 1\\ 
\hline
DOC $(A_{17})$ & CO$_2$ injected into the geological stock & 1 & 0\\ 
\hline
\end{tabular}
\caption{Parameter Values of the Current RNDCR Framework}
\label{table:parameter}
\end{table} 

\subsubsection{Kinetics of the emission reaction $R_5$}\label{emission}

In this section, we discuss the kinetics of the emission reaction $A_4\rightarrow A_2$ in the single CDR technology case in detail and then indicate briefly how the results extend to the general case.
In the single  CDR technology case, the emission kinetics 
$$K_5(A_1,A_2,A_3,A_4,A_i) = k_5(A_4-((1- \lambda_i)+ \mu_i\lambda_i)A_i).$$
Clearly, in general, the kinetics is not in power law form. To ensure this, we apply a power law approximation: this can be done in the 2-variable form recalled below from the Appendix of \cite{FMRL2018} since the kinetics depends only on the variables $A_4$ and $A_i.$

To keep the discussion here simple, consider a rate function that depends on two variables, say $V(X_1, X_2)$. To compute the power-law approximation, we start by choosing an operating point, such as a steady state. linearization in logarithmic coordinates reveals that $V$ can be approximated by a product of power-law functions of the form
$$V=\alpha X_1^pX_2^q.$$

The kinetic orders $p$ and $q$ are the slopes of the approximating function in logarithmic coordinates and are computed through partial derivatives of $V$ with respect to $X_1$ or $X_2$, respectively. They are given by 
$$p=\dfrac{\partial V}{\partial X_1}\cdot \dfrac{X_1}{V}$$
and 
$$q=\dfrac{\partial V}{\partial X_2}\cdot \dfrac{X_2}{V}.$$

Applying this to $K_5$ and denoting with $e_i$ the kinetic order of $A_4$ and $f_i$ the kinetic order of $A_i$ in the power law approximation, we obtain the following proposition: 
\\
\begin{proposition}\label{opp point}
For any positive steady state as operating point with coordinates $(A_{4_0}, A_{i_0}),$
\begin{enumerate}
    \item [(i)] $e_i \geq 1$ and $f_i \leq 0$
    \item [(ii)] $e_i + f_i =1$
    \item [(iii)]  $e_i=1, f_i = 0$ iff emission is linear mass action iff $\lambda_i= 1, \mu_i = 0$.
\end{enumerate}
\end{proposition}
\begin{proof}
    To prove (i) and (ii), apply the formulas for  $p$ and $q$ to $$V = k_5(A_{4_0}-((1-\lambda_i)+\mu_i\lambda_i))A_{i_0})$$ at the given (symbolic) operating point, we obtain 
    \begin{align*}
        e_i & = \dfrac{A_{4_0}}{A_{4_0}-[(1 - \lambda_i)+ \mu_i\lambda_i)]A_{i_0}}; \text{ and} \\
         f_i &= \dfrac{[(1-\lambda_i)+\mu_i \lambda_i)]A_{i_0}}{A_{4_0}-[(1-\lambda_i) +\mu_i\lambda_i)]A_{i_0}}.   
         \end{align*}
    Similarly, the equivalences in (iii) are straightforward calculations.
\end{proof}

\begin{example}
As examples for (ii), we have $\mu_i = 0$ and $\lambda_i= 1$ for DAC, BECCS, OF and DOC (Direct Ocean Capture).
\end{example}

\subsubsection{Structure of the kinetic order matrix}\label{structure}

In the single CDR case, the first 6 rows of the kinetic order matrix are as follows:

\begin{table}[h]
\centering
\begin{tabular}{|c|c|c|c|c|c|}
\hline
 Reaction/Species & $A_1$ & $A_2$ & $A_3$& $A_4$ & $A_i$ \\
\hline
 $R_1$ & $p_1$ & $q_1$ & 0 & 0 & 0 \\
\hline
$R_2$ & $p_2$ & $q_2$ & 0 & 0 & 0 \\
\hline
$R_3$ & 0 & 1 & 0 & 0 & 0 \\
\hline
$R_4$ & 0 & 0 & 1 & 0 & 0 \\
\hline
$R_5$ & 0 & 0 & 0 & $e_i$ & $1-e_i$ \\
\hline
$R_{i,6}$ & 0 & 0 & 0 & 0 & 1 \\
\hline
\end{tabular}
\end{table} 

\noindent The remaining rows are CDR-specific. In current models, most of the remaining reactions have mass action kinetics.

\subsection{Classes and Anderies decomposition of an RNCDR system}\label{classes}

The classes of Anderies systems can be extended to an RNCDR system in the following ways.
Let $R$, $Q$ be the ratios defined by the kinetic orders of the Anderies subnetwork. The RNCDR system is \textbf{positive}, \textbf{negative} and $\bm{P}$\textbf{-null} if $R$ is either positive, negative or 0,  accordingly. The system is $\bm{Q}$\textbf{-null} if $Q = 0.$   

Like any subnetwork, the Anderies subnetwork induces together with its (graph) complement a binary decomposition of the RNCDR system. If this decomposition is independent, properties of the Anderies subnetwork has a considerable impact on those of the whole system. We illustrate this effect by reviewing the results of Fortun et. al \cite{FLMR2024} on the DAC system and extending them to the new class of $Q$-null DAC systems.
We have the following analogous results for $Q$-null Anderies systems:
\\
\begin{proposition}
Any $DAC_{Q_0}$ system is monostationary and has ACR only in $A_1$.
\end{proposition}  

\begin{proof}
For $DAC_{Q_0}$, we have a parametrization extending that of the Anderies subnetwork as follows: 
$A_1 = \left(\dfrac{k_1}{k_2}\right)^{\frac{1}{p_2-p_1}}$, $A_2$ a free parameter, $A_3 = \left(\dfrac{1}{\beta}\right) A_2$, $A_4= \left(\dfrac{k_4}{k_5}\right)A_2$, and $A_5= \left(\dfrac{k_4}{k_6}\right)A_2$. We obtain the following expression for $A_2$:
$$A_2\left(1+\dfrac{1}{\beta}+\dfrac{k_4}{k_5}+\dfrac{k_4}{k_6}\right)=A_0-A_1,$$
so that an analogous argument as in the Anderies case shows the monostationarity.

On the other hand, in view of its Anderies decomposition, the equilibria set of a DAC system is a subset of the equilibria set of its Anderies subnetwork. It follows that a DAC system is also a PLP system and $\widetilde{S}^\perp = \text{span } \{ [-Q, 1, 1,1,1 ]^\top \}$. For $DAC_{Q_0}$ systems, $Q = 0,$ so it follows from the Species Hyperplane Criterion that $A_1$ is an ACR species while $A_2, A_3,A_4$ and $A_5$  are not.   
\end{proof}

The following table updates the summary of DAC system properties from \cite{FLMR2024}:

\begin{center}
\small
\begin{tabular}{ |c|c| } 
\hline
Kinetic Property & DAC system\\
\hline
Existence of at least one equilibrium & True for positive, negative and null systems\\
\hline
\multirow{3}*{Capacity for multiple steady states} & Positive DAC: Multistationary \\ 
& Negative DAC: Contains monostationary systems  \\ 
& Null DAC: Monostationary \\ 
\hline
\end{tabular}
\end{center}

\section{RNCDR model of a BECCS system }

This section outlines the results of our model construction and analysis of the negative emission technology known as Bioenergy with Carbon Capture and Storage (BECCS). Alongside afforestation/reforestation (AR), BECCS is widely recognized as one of the most mature NETs.  This is because the mature methods of bioenergy production are supplemented by techniques of carbon capture and storage.  Most pathways/scenarios for achieving the goal of limiting global temperature rise to a maximum of $1.5^\circ$ C consider only BECCS, a limitation which has been criticized by some environmental researchers \cite{STREF021}. In any case, this fact emphasizes the importance of a thorough analysis of the model.

\subsection{Model and basic properties}
The global carbon cycle model dynamics discussed here are based on the following ODE system model.
\begin{equation}\label{eq:BECCS ODEs}
\left.
 \begin{array}{cl}
\dot{A_1}&=k_1A_1^{p_1}A_2^{q_1} - k_2A_1^{p_2}A_2^{q_2} - k_6A_1 \\
\dot{A_2 }&= k_2A_1^{p_2}A_2^{q_2} - k_1A_1^{p_1}A_2^{q_1} +k_5A_4 -a_mA_2 +a_m\beta A_3  \\
\dot{A_3} &=  a_mA_2 - a_m\beta A_3 \\
\dot{A_4} &= k_7A_8 - k_5A_4\\
\dot{A_8} &= k_6A_1 - k_7A_8
 \end{array}
 \right \}
\end{equation}

\noindent This model extends the power-law system derived by Fortun et al. \cite{FMRL2018} for Anderies et al.'s pre-industrial carbon cycle \cite{ANDE2013}. The original three-box model by Anderies et al. considers carbon interactions in the land-atmosphere-ocean system (designated as $A_1$, $A_2$, and $A_3$, respectively) along with industrial carbon transfer processes like fossil fuel combustion, resulting in the linear transfer of carbon geological stock ($A_4$) to the atmosphere. The expanded model, illustrated in Figure \ref{fig:beccs}, integrates BECCS  by introducing an additional box ($A_8$) to store carbon directly sequestered from the atmosphere, with the transfer rate assumed to be linear. The system also accounts for a potential ``leak" (from $A_8$ to $A_4$) to evaluate the CDR performance, even when a leak is present.

\begin{figure}[t]%
\centering
\includegraphics[width=0.8\textwidth]{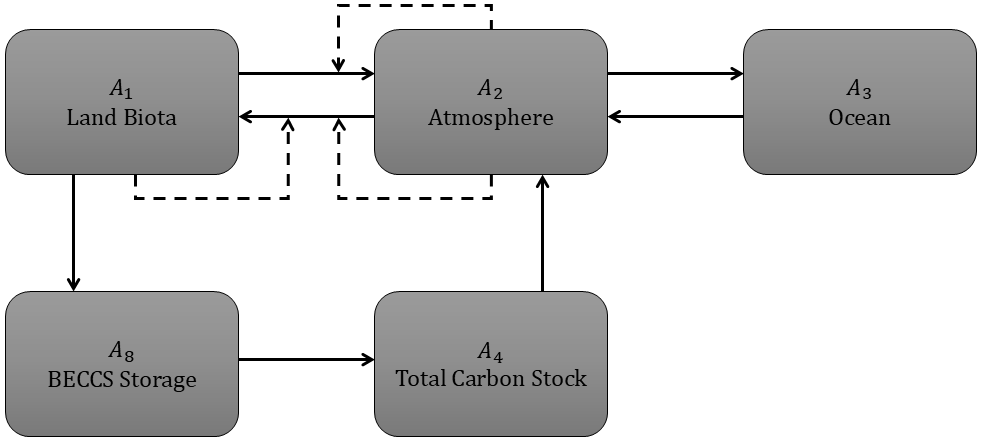}
\caption{A carbon cycle with BECCS.} \label{fig:beccs}
\end{figure}

\subsubsection{A power law kinetic representation of a BECCS system}

The analysis of the carbon cycle with BECCS begins by constructing its power law kinetic representation. This representation, referred to as \textbf{BECCS system} hereafter, involves a CRN with power law kinetics that is dynamically equivalent to the given system; i.e., they have identical ODE system (\ref{eq:BECCS ODEs}).  

The desired CRN is set up by extending the deficieny-zero CRN of the pre-industrial Anderies system in Eq. (\ref{eq:Anderies2}) of Section \ref{section: CRN}. The complete CRN representation of BECCS is a network with 5 species, 7 complexes, and 7 reactions:

\begin{equation*}
\begin{blockarray}{clcc}
\text{Carbon transfer} & \text{Reaction} & \text{Reaction vector} & \text{Reaction rate}  \\
A_2 \to A_1 & R_1: A_1+2A_2 \xrightarrow{k_1} 2A_1+A_2 &  [1, -1, 0, 0, 0]^\top  & k_1 A_1^{p_1}A_2^{q_1}\\
A_1 \to A_2 & R_2: 2A_1+A_2 \xrightarrow{k_2} A_1+2A_2 &  [-1, 1, 0, 0, 0]^\top & k_2 A_1^{p_2}A_2^{q_2}\\
A_2 \to A_3 & R_3: A_2 \xrightarrow{a_m} A_3  &  [0, -1, 1, 0, 0]^\top & a_m A_2\\
A_3 \to A_2 & R_4: A_3 \xrightarrow{a_m\beta} A_2 &  [0, 1, -1, 0, 0]^\top  & a_m \beta A_3\\
A_4 \to A_2 & R_5: A_4 \xrightarrow{k_5} A_2 &  [0, 1, 0, -1, 0]^\top & k_5 A_4 \\
A_1 \to A_8 & R_{8,6}: A_1 \xrightarrow{k_6} A_8  &  [-1, 0, 0, 0, 1]^\top & k_6 A_1\\
A_8 \to A_4 & R_{8,7}: A_8 \xrightarrow{k_7} A_4 &  [0, 0, 0, 1, -1]^\top  & k_7 A_8 \\
\end{blockarray}
\end{equation*}

There are two linkage classes $\ell=2$ which are $\mathscr{L}_1= \{ A_1+2A_2 \rightleftarrows 2A_1 + A_2\}$ and $\mathscr{L}_2= \{ A_1\to A_5 \to A_4 \to A_2 \rightleftarrows A_3\}$. The network's rank $s=4$ since
$$ S =\text{span} \left\lbrace   \begin{bmatrix} 1 \\ -1 \\ 0 \\ 0 \\ 0 \end{bmatrix},\begin{bmatrix} 0 \\ 1 \\ -1 \\ 0 \\ 0 \end{bmatrix},\begin{bmatrix} 1 \\ 0 \\ 0 \\ 0 \\ -1 \end{bmatrix},\begin{bmatrix} 0 \\ 0 \\ 0 \\ 1 \\ -1 \end{bmatrix} \right\rbrace. $$
Hence, its deficiency $\delta=n-\ell-s=7-2-4=1$. Table \ref{table:beccsnumbers} lists the network numbers for the BECCS system, obtained from the CRNToolBox \cite{CRNTool}. This information includes details on the number of species, complexes, reactions, linkage classes, rank, and deficiency.

\begin{table}[h] 
    \centering
    \begin{tabular}{|c|c|c|}
    \hline
    Network number & Symbol & BECCS System \\
    \hline
    Species & $m$ & 5 \\
    \hline
    Complex & $n$ & 7 \\
    \hline
    Reactant complexes & $n_r$ & 7 \\
    \hline
    Reactions & $r$ & 7 \\
    \hline
    Irreversible reactions & $r_{irr}$ & 3 \\
    \hline
    Linkage classes & $l$ & 2\\
    \hline
    Strong linkage classes & $sl$ & 1\\
    \hline
    Terminal strong linkage classes & $t$ & 2\\
    \hline
    Rank & $s$ & 4 \\
    \hline
    Reactant rank & $q$ & 5\\
    \hline
    Deficiency & $\delta$ & 1\\
    \hline
    Reactant deficiency & $\delta_\rho$ & 2\\
    \hline
\end{tabular}
\caption{Network numbers of BECCS system from CRNTToolBox \cite{CRNTool}.}
\label{table:beccsnumbers}
\end{table} 

Additionally, Table \ref{table:beccsproperties} presents various properties of the BECCS system, derived from the findings in CRNToolBox \cite{CRNTool}.

\begin{table}[h] 
    \centering
    \begin{tabular}{|c|c|}
        \hline
        Property & Report \\
        \hline
        Concordance &  Discordant \\
        \hline
        Strong concordance & Not strong concordant\\
        \hline
        Conservative & Yes \\
        \hline
        Independent linkage classes & No\\
        \hline
        Positive dependent & Has positively dependent  reaction vectors \\
        \hline
        Positive dependent & Has positively dependent  reaction vectors \\
        \hline
        Regular & Yes\\
        \hline
    \end{tabular}
    \caption{Some properties of BECCS system from CRNTToolBox}
    \label{table:beccsproperties}
\end{table}

The power law kinetics of BECCS is encoded in its corresponding kinetic order matrix:

\begin{equation*}
F= \begin{blockarray}{cccccl}
A_1 & A_2 & A_3 & A_4 & A_5 & \\
\begin{block}{[ccccc]l}
p_1 & q_1 & 0 & 0 & 0 & R_1 \\
p_2 & q_2 & 0 & 0 & 0 & R_2 \\
0 & 1 & 0 & 0 & 0 & R_3 \\
0 & 0 & 1 & 0 & 0 & R_4 \\
0 & 0 & 0 & 1 & 0 & R_5 \\
1 & 0 & 0 & 0 & 0 & R_{8,6} \\
0 & 0 & 0 & 0 & 1 & R_{8,7} \\
\end{block}
\end{blockarray}
\end{equation*}

\subsection{ACR analysis of BECCS}

A network decomposition is said to be \textbf{independent} if its stoichiometric subspace is a direct sum of the subnetwork stoichiometric subspaces. Feinberg \cite{FEIN1987} showed that, in an independent decomposition, the intersection of the set of steady states of the subnetworks is identical to the set of steady states of the main network. The finest independent decomposition (FID) of BECCS contains two subnetworks:
\begin{align*}
\mathscr{P}_1 & =\{ A_1+2A_2 \rightleftarrows 2A_1 + A_2, \; A_1\to A_5 \to A_4 \to A_2 \} \\
\mathscr{P}_2 & =\{ A_2\rightleftarrows A_3 \}  
\end{align*}
The ODE system of $\mathscr{P}_1$ is given by:
\begin{equation*}
\left.
 \begin{array}{cl}
\dot{A_1}&=k_1A_1^{p_1}A_2^{q_1} - k_2A_1^{p_2}A_2^{q_2} - k_6A_1 \\
\dot{A_2 }&= k_2A_1^{p_2}A_2^{q_2} - k_1A_1^{p_1}A_2^{q_1} +k_5A_4  \\
\dot{A_4} &= k_7A_8 - k_5A_4\\
\dot{A_8} &= k_6A_1 - k_7A_8
 \end{array}
 \right \}
\end{equation*}
Setting the equations in the ODE system to 0, we find the following solutions which are partly parametrized by $A_1$. That is, if $A_1=\tau_1>0$,
$$
    A_1 = \tau_1, \quad
    A_2 = \text{Root of } \left\lbrace k_1\tau ^{p_1}A_2^{q_1} - k_2\tau_1^{p_2}A_2^{q_2} - k_5\tau_1 = 0\right\rbrace, \quad
    A_4  = \dfrac{k_5}{k_4}\tau_1, \quad
    A_8 = \dfrac{k_5}{k_6}\tau_1.
$$
On the other hand, the ODE system of $\mathscr{P}_2$ is
\begin{equation*}
\left.
 \begin{array}{cl}
\dot{A_2 }&=  -a_mA_2 +a_m\beta A_3  \\
\dot{A_3} &=  a_mA_2 - a_m\beta A_3 \\
 \end{array}
 \right \}
\end{equation*}
Clearly, at steady state, $A_3=\dfrac{1}{\beta}\tau_2$   if $A_2=\tau_2>0$.
Hence, the steady state of the entire network can be obtained by ``stitching'' the equilibria set of $\mathscr{P}_1$ and $\mathscr{P}_2$: 
\begin{align*}
     A_1 &= \tau_1, \quad
    A_2 =  \tau_2 = \text{Root of } \left\lbrace k_1\tau_1^{p_1}A_2^{q_1} - k_2\tau_1^{p_2}A_2^{q_2} - k_5\tau_1 = 0\right\rbrace  \\
    A_3  &= \dfrac{1}{\beta}\tau_2, \quad
    A_4  = \dfrac{k_5}{k_4}\tau_1, \quad
    A_8 = \dfrac{k_5}{k_6}\tau_1.
\end{align*}

The following result presents a particular condition under which a $P$-null system will demonstrate ACR in both $A_2$ and $A_3$.
\\
\begin{proposition} \label{prop:ACRBECCS}
Let $(\mathscr{N},K)$ be a BECCS system. 
\begin{enumerate} 
    \item A P-null BECCS system with $p_1=p_2=1$ has $A_2$ and $A_3$ as its only ACR species. 
    \item All other BECCS systems have no ACR species.
\end{enumerate}
\end{proposition}

\begin{proof}
From the equilibria set described above, we observe that since $\tau_1$ is a free parameter, $A_1$ does not have ACR. Consequently, $A_4$ and $A_8$ are also not ACR species. 
We now consider the equation for $A_2$. By dividing by $A_1=\tau_1$, we obtain 
\begin{align*}
k_1\tau_1^{p_1-1}A_2^{q_1} - k_2\tau_1^{p_2-1}A_2^{q_2} &=k_5 \\
\Rightarrow k_1\tau_1^{p_2-1}\left( \tau_1^{p_1-p_2}A_2^{q_1}-k_2A_2^{q_2}\right) &=k_5
\end{align*} 
In case of (i), the $\tau_1$ terms are equal to  1, and the LHS depends only on $A_2$ and constants while the RHS is a constant. The non-invariance of $A_2$ immediately leads to a contradiction. This shows that $A_2$ and $A_3$ are ACR species. In case of (ii), by taking logarithms, we see that at least one $\tau_1$ term varies. If we assume $A_2$ is a constant at equlibria, i.e. has ACR, we have a varying LHS (with $\tau_1$) and a constant RHS, a contradiction. Hence, there are no ACR species in these remaining cases.

\end{proof}

\subsection{Multistationarity analysis of BECCS systems}

A CRN with stoichiometric matrix $N$ is said to be \textbf{injective} if for any distinct stoichiometrically compatible species vectors $x$ and $y$, we have $NK(x) \neq NK(y)$ for all kinetics $K$ endowed on the CRN. Note that if a CRN is injective, then it is monostationary. In other words, an injective CRN cannot support multiple positive steady states for any rate constants. Wiuf and Feliu \cite{WIUF2013, FELIU2013} established a criterion to identify if a system is injective. 

Let $M=N \; \text{diag}(z) \; F \; \text{diag}(k)$, where $N$ represents the stoichiometric matrix and $F$ is the kinetic order matrix of the PLK system. Consider the symbolic matrix $M^*$ created by using symbolic vectors $k=(k_1,\dots,k_m)$ and $z=(z_1,\dots,z_r)$. Assume $\{ \omega^1,\dots,\omega^d \}$ forms a basis of the left kernel of $N$, and $i_1,\dots,i_d$ represent row indices. Define the $m \times m$ matrix $M^*$ by substituting the $i_j$-th row of $M$ with $\omega^j$. 
\\
\begin{theorem}(\cite{WIUF2013,FELIU2013})\label{theorem:wiuffeliu}
The interaction network with power law kinetics and fixed kinetic orders is injective if and only if the determinant of $M^*$ is a non-zero homogeneous polynomial with all coefficients being positive or all being negative. 
\end{theorem}
$ $
\\
We apply this theorem to BECCS:
\\
\begin{proposition}
The BECCS system is injective, and hence monostationary, if any of the following cases hold:
\begin{enumerate}
\item[(i)] $p_1<0$, $p_2>0$, $q_1>0$, and $q_2<0$; 
\item[(ii)] $p_1=p_2=0$, $q_1>0$, and $q_2<0$; or 
\item[(iii)] $q_1=q_2=0$, $p_1<0$, and $p_2>0$. 
\end{enumerate}
\noindent For all other cases, the network is not injective.
\end{proposition}

\begin{proof}
Using the computational approach and Maple script provided by the authors in \cite{FELIU2013}, we obtain the determinant of $M^*$ for the BECCS system: 

\begin{align*}
det (M^*)=&-p_1k_1k_2k_4k_5z_1z_3z_5z_7 – p_1k_1k_3k_4k_5z_1z_4z_5z_7 +p_2k_1k_2k_4k_5z_2z_3z_5z_7 \\
&+ p_2k_1k_3k_4k_5z_2z_4z_5z_7 +q_1k_1k_2k_3k_4z_1z_4z_5z_6 + q_1k_1k_2k_3k_5z_1z_4z_6z_7 \\
&+q_1k_2k_3k_4k_5z_1z_4z_5z_7 -q_2k_1k_2k_3k_4z_2z_4z_5z_6 -q_2k_1k_2k_3k_5z_2z_4z_6z_7 \\
&-q_2k_2k_3k_4k_5z_2z_4z_5z_7 +k_1k_2k_4k_5z_3z_5z_6z_7+k_1k_3k_4k_5z_4z_5z_6z_7
\end{align*}

\noindent Since the last two terms of the determinant are always positive, the system is injective for the three cases specified in the proposition. In all other cases, the network is not injective.
\end{proof}

\begin{remark}
If a CRN is not injective, multistationarity does not necessarily follow.    
\end{remark}
$ $
\\
In addition, the CRN representation of BECCS is a ``regular" (in the sense of Feinberg \cite{FEIN1995DOA}), non-weakly reversible network with a deficiency of one. This network is suitable for applying the Deficiency-One Algorithm (DOA) for power-law kinetic systems described in \cite{FMRL2018} to assess its ability to support multiple steady states. Fundamentally, the DOA translated the problem of determining a system's capacity to admit multiple steady states -- a non-linear problem -- into a problem involving a linear system of equations and inequalities. Applying the DOA to BECCS shows that there is a set of non-negative kinetic orders, $p_1, p_2, q_1$, and $q_2$, such that the system has the capacity to admit multiple steady states: 
\\
\begin{proposition} \label{prop:DOA BECCS}
There exist non-negative values for $p_1, p_2, q_1$, and $q_2$ such that BECCS is multistationary.
\end{proposition}
$ $
\\
The above result is obtained by implementing the DOA on the BECCS system. Appendix \ref{appendix:BECCS} provides the complete proof.

\begin{table}[h]
    \centering
\begin{adjustbox}{width=1\textwidth}
\small
  \begin{tabular}{|l|l|l|}
    \hline
        \textbf{BECCS class} & \textbf{Mono-/multi-stationary subsets} & \textbf{Injective/non-injective subsets} \\
    \hline
        BECCS positive &  Multistationary: $p_1>p_2>1$ and $q_1>q_2>0$  &  Non-injective for all values of kinetic orders \\
    \hline
        BECCS negative &  Monostationary: $p_1<0$, $p_2>0$, $q_1>0$, and $q_2<0$ & Injective: $p_1<0$, $p_2>0$, $q_1>0$, and $q_2<0$  \\
    \hline
        BECCS $P$-null & Monostationary: $p_1=p_2=0$, $q_1>0$, and $q_2<0$  &  Injective: $p_1=p_2=0$, $q_1>0$, and $q_2<0$ \\
    \hline
        BECCS $Q$-null & Monostationary: $q_1=q_2=0$, $p_1<0$, and $p_2>0$ & Injective: $q_1=q_2=0$, $p_1<0$, and $p_2>0$ \\
    \hline
\end{tabular}
\end{adjustbox}
    \caption{Summary of mutistationarity and injectivity analysis for BECCS }
    \label{table: BECCS summary}
\end{table}
Table \ref{table: BECCS summary} summarizes the pertinent results that link the signs of $R$ and $Q$ to the dynamic properties of the corresponding BECCS system.

\subsection{Virtually complex balanced equilibria of a class of BECCS systems}

In this section, we introduce a special class of BECCS systems with the property that the stoichiometric subspace and kinetic flux subspace coincide. We show that for such systems, there is at least one complex balanced equilibrium in each positive stoichiometric class. We also derive a necessary condition for atmospheric carbon dioxide reduction for this special case.

\subsubsection{ A balanced negative BECCS system}
Since $n_r= n = r = 7$, any BECCS network is cycle terminal and non-branching implying it is always a PL-RDK system when endowed with power law kinetics. Moreover, its kinetic flux subspace is well-defined. Specifically, 
$$ \widetilde{S} =\text{span} \left\lbrace  \begin{bmatrix} p_2-p_1 \\ q_2-q_1 \\ 0 \\ 0 \\ 0 \end{bmatrix},  \begin{bmatrix} 0 \\ 1 \\ -1 \\ 0 \\ 0 \end{bmatrix},  \begin{bmatrix} 1 \\ 0 \\ 0 \\ 0 \\ -1 \end{bmatrix}, \begin{bmatrix} 0 \\ 0 \\ 0 \\ 1 \\ -1 \end{bmatrix}, \begin{bmatrix} 0 \\ 1 \\ 0 \\ -1 \\ 0 \end{bmatrix} \right\rbrace.$$

If $q_2 -q_1 \neq 0$, we can write the first vector as $[R \;1 \;0 \;0 \;0]^\top$, where, as usual,  $R = \dfrac{p_2-p_1}{q_2-q_1}$. We call a BECCS system \textbf{balanced negative} if $R=-1$. 
\begin{remark}
A BECCS system is balanced negative if and only if $S=\widetilde{S}$. Furthermore, $\dim \widetilde{S}=4$ and $(\widetilde{S})^\perp =\text{span } \{[1 \; 1\; 1\; 1\; 1]^\top \}$. 
\end{remark}

\subsubsection{Virtually complex balanced equilibria of a balanced negative BECCS system }
To describe the equilibria of a balanced negative BECCS system, we construct a weakly reversible transform $(\mathscr{N}^\#, K^\#)$ of the system. We denote the stoichiometric subspace and kinetic order of the transform by $S^\#$ and $\widetilde{S}^\#$, respectively.
\\
\begin{proposition}
Any CBECCS system $(\mathscr{N},K)$ has a weakly reversible network transform $(\mathscr{N}^\#, K^\#)$ such that $S^\# = S$ and $\widetilde{S}^\#=\widetilde{S}$.
\end{proposition}
\begin{proof}
We construct a weakly reversible network transform that is dynamically equivalent to CBECCS through the following steps:
\begin{enumerate} 
\item The reaction in the linkage class $A_1 \to A_8 \to A_4 \to A_2$ are shifted by the complex $A_1+A_2$. 
\item The reaction $A_1 +2A_2 \to A_2 + 2A_1$ is divided into two reactions using the sum of the rate constants $k_1=\frac{k_1}{2}+\frac{k_1}{2}$. One of the reactions is merged to the shifted linkage class to complete a cycle.
\item The reversible pair $A_1+ 2A_2 \rightleftarrows A_2+2A_1$ is shifted by the complex $A_4$. 
\item The reversible pair $A_2 \rightleftarrows A_3$ is left unchanged.
\end{enumerate}
This results in the following weakly reversible network:
\begin{align*}
2A_1 + A_2 &\to A_1+A_2+A_8 \\
A_1 + 2A_2 &\to 2A_1 + A_2 \\
A_1 + A_2 + A_8 &\to A_1+A_2+A_4 \\
A_1+A_2+A_4 &\to A_1+2A_2 \\
A_1+2A_2+A_4 &\rightleftarrows 2A_1+A_2+A_4 \\
A_2 &\rightleftarrows A_3
\end{align*}
It is non-branching ($n_r=r=8$) and hence any power-law kinetics is PL-RDK. Since shifting and splitting of reactions are $S$-invariant transformations, $S^\#=S$. Moreover, the kinetic order matrix of the transform differs from the original network only in the addition of a copy of the first row. Hence, the generated subspaces of fluxes remain the same.
\end{proof}

\begin{corollary}
For a balanced negative BECCS system, all four subspaces $S, \widetilde{S}, S^\#,$ and $\widetilde{S}^\#$ coincide.
\end{corollary}
\begin{proof}
The claim holds because for a balanced negative BECCS system, $S=\widetilde{S}$.
\end{proof}

In a balanced negative BECCS system, we have the following result:
\\
\begin{proposition}\label{prop:vcb}
For any balanced negative BECCS system, there are rate constants such that the set 
$$VCB = \left\lbrace \left. \left[   e^\alpha \; e^\alpha \; e^\alpha \; e^\alpha \; e^\alpha \right]^\top \right| \alpha \in \mathbb{R} \right\rbrace$$
are positive equilibria, called \textbf{virtually complex balanced equilibria}.
\end{proposition}

\begin{proof}
The system's network transform has a kinetic deficiency $=8-3-4=1$. Hence, by Theorem 1 of Müller and Regensburger \cite{MURE2014}, there are rate constants such that weakly reversible transform is complex-balanced. It follows further that for these rate constants, the network transform's set of complex balanced equilibria is given by $$ Z_+ (\mathscr{N}^\#, K^\#)= \left\lbrace x \in \mathbb{R}_>^\mathscr{S} | \ln x - \ln x^* \in \alpha [1 \; 1 \; 1 \; 1\; 1 ]^\top \right\rbrace.$$ Since $[1 \; 1 \; 1 \; 1\; 1 ]^\top $ is an equilibrium, we set it as a reference point and obtain for any VCB equilibrium, $x=\left[   e^\alpha \; e^\alpha \; e^\alpha \; e^\alpha \; e^\alpha \right]^\top$.
\end{proof}

\begin{corollary}
Any balanced negative BECCS system has no ACR species.
\end{corollary}
\begin{proof}
Since the exponential map is a strictly monotonic map $\mathbb{R}\to \mathbb{R}_{>0}$, for $\alpha_1 > \alpha_2$,   $e^{\alpha_1}- e^{\alpha_2}>0$. This shows that there are no ACR species.
\end{proof}

\begin{corollary}
There is exactly one VCB equilibrium in each positive stoichiometric class $\neq S$ of a balanced negative BECCS system.
\end{corollary}
\begin{proof}
The positive stoichiometric classes are parametrized $[k \; k \; k \; k\; k ]^\top $ with $k \geq 0$. For $k > 0$, choose $\alpha = \ln k$, then $e^\alpha$ is contained in $[k \; k \; k \; k\; k ]^\top + S$. The difference of the two VCB is still a VCB and hence in $S^\perp$. If it were in $S$, then it would be zero. 
\end{proof}

\subsubsection{A necessary condition for carbon dioxide reduction by a balanced negative BECCS system}

Finally, we discuss a necessary condition for atmospheric carbon dioxide reduction in a balanced negative system. The following finding indicates that atmospheric carbon reduction occurs when the initial value of $A_2$ exceeds its steady-state value. Additionally, it is important to note further that an initial condition $x^0 =\left[ A_1^0 \;\; A_2^0 \;\; A_3^0 \;\; A_4^0 \;\; A_8^0 \right]^\top$ of a BECCS system determines a unique stoichiometric class $S^0$. 
\\
\begin{proposition}
Suppose there is a balanced negative BECCS system with a VCB equilibrium in $S^0\neq S$  where $A_2^0 > A_2$; that is, $A_2=\xi A_2^0, \; 0 < \xi <1.$ Then at this equilibrium point, $A_2$ cannot drop by $80 \%$ or more of its original value.
\end{proposition}

\begin{proof}
From Proposition \ref{prop:vcb}, we have $\xi A_2^0=e^\alpha$. Since the system is conservative, the sum of initial condition values must be the same as the sum of the positive equilibrium quantities. That is,
\begin{align*}
A_1^0+ A_2^0+ A_3^0+ A_4^0+ A_8^0 &= 5e^\alpha =5 \left( \xi A_2^0 \right) \\
\Leftrightarrow A_1^0+ A_3^0+ A_4^0+ A_8^0 &= \left( 5\xi -1\right)A_2^0.
\end{align*}
Since the LHS of the last equation is positive, the condition implies in particular $\xi > \frac{1}{5}$.

\end{proof}

\section{RNCDR model of an AR system without natural reforestation } \label{section:AR system}
In this section, we introduce a model of \textbf{afforestation and reforestation (AR)} as an RNCDR kinetic system. AR storage involves planting trees on land that was not previously forested. These trees absorb CO$_2$ from the atmosphere and help restore carbon through the natural photosynthesis and respiration processes of the land. However, when one reforests or afforests, the natural reforestation or growth of trees in small patches of areas where there are adjacent mother trees is evident, therefore, increasing the photosynthesis rate over time. In this study, we consider the case where natural reforestation is absent.  

\subsection{Model and basic properties of an AR system}

Without natural reforestation, the schematic diagram of AR storage is shown in Figure \ref{AR}. In this model, we add the total carbon stock ($A_4$) in the Anderies pre-industrial system so that it considers the fossil fuel combustion, which is a transfer of carbon from the $A_4$ to $A_2$.

We introduce A$_{15}$ as the afforestation and reforestation storage which restores CO$_2$ from the land respiration process. We also consider the possibility of ``leak'' of the AR storage to the total carbon stock, resulting to the carbon transfer from A$_{15}$ to A$_4$. Lastly, the requirement of the kinetics $K_5$ of carbon transfer from $A_4$ to $A_2$ requires the $A_{15}$ in the general case of its power law form. To satisfy the positivity condition of AR system, we shall add an influence arrow from A$_{15}$ to the reaction $A_4\to A_2$ and shift this reaction by $A_{15}$.

\begin{figure}[h]
\centering
\includegraphics[scale=0.4]{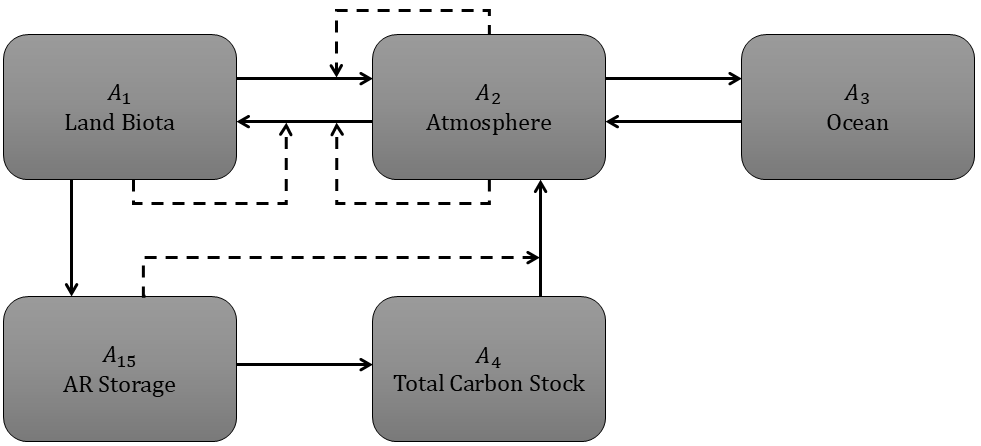}
\caption{The schematic diagram of AR system}\label{AR}
\end{figure}

The table below shows the reactions of the AR system based on the carbon transfer from the schematic diagram and their corresponding power-law kinetics. The first four reactions are drawn from the Anderies system and its emission reactions that will produce the same kinetics. The emission reaction of $R_5$ will be obtained using $\mu_i=0$ and $\lambda_i=0.5$ in the table presented in Section \ref{parameters}. Thus, the kinetics on $R_5$ will be $k_5 (A_4-0.5A_{15})$ and its power-law approximation will be $k_5 A_4^{e_{15}} A_{15}^{f_{15}}$. The rest of the reactions will be mass action kinetics.

\begin{table}[h]
    \centering
    \begin{tabular}{|c|c|c|}
    \hline
        Reaction symbol & Reactions and rate constants & Corresponding power-law kinetics \\
    \hline
        $R_1$ & $A_1 + 2A_2 \xrightarrow{k_1} 2A_1 + A_2$ & $k_1 A_1^{p_1} A_2^{q_1}$ \\
    \hline
        $R_2$ & $A_2 + 2A_1 \xrightarrow{k_2} A_1 + 2A_2$ & $k_2 A_1^{p_2} A_2^{q_2}$ \\
        \hline
        $R_3$ & $A_2 \xrightarrow{a_m} A_3$ & $a_m A_2$ \\
        \hline
        $R_4$ & $A_3 \xrightarrow{a_m\beta} A_2$ & $a_m\beta A_3$\\
        \hline
        $R_5$ & $ A_4 +A_{15} \xrightarrow{k_5} A_2 +A_{15} $ & $k_5 A_4^{e_{15}}A_{15}^{f_{15}}$\\
        \hline
        $R_{15,6}$ & $A_1\xrightarrow{k_{15,6}} A_{15}$ & $k_{15,6}A_1$\\
        \hline
        $R_{15,7}$ & $A_{15}\xrightarrow{k_{15,7}} A_{4}$ & $k_{15,7}A_{15}$\\
        \hline
        
    \end{tabular}
    \caption{CRN representation of AR system and the corresponding power-law kinetcics}
    \label{AR Power Law}
\end{table}

This reactions will lead to the following network numbers and reports from CRNT Toolbox.
\begin{table}[h] 
    \centering
    \begin{tabular}{|c|c|c|}
    \hline
    Network number & Symbol & AR System \\
    \hline
    Species & $m$ & 5 \\
    \hline
    Complex & $n$ & 9 \\
    \hline
    Reactant complexes & $n_r$ & 7 \\
    \hline
    Reactions & $r$ & 7 \\
    \hline
    Irreversible reactions & $r_{irr}$ & 3 \\
    \hline
    Linkage classes & $l$ & 4\\
    \hline
    Strong linkage classes & $sl$ & 7\\
    \hline
    Terminal strong linkage classes & $t$ & 4\\
    \hline
    Rank & $s$ & 4 \\
    \hline
    Reactant rank & $q$ & 5\\
    \hline
    Deficiency & $\delta$ & 1\\
    \hline
    Reactant deficiency & $\delta_\rho$ & 2\\
    \hline
\end{tabular}
\caption{Network numbers of AR system from CRNTToolBox}
\label{table:arnumbers}
\end{table} 

\begin{table}[h] 
    \centering
    \begin{tabular}{|c|c|}
        \hline
        Property & Report \\
        \hline
        Concordance &  Discordant \\
        \hline
        Strong concordance & Not strongly concordant\\
        \hline
        Conservative & Yes \\
        \hline
        Independent linkage classes & No\\
        \hline
        Positive dependent & Has positively dependent  reaction vectors \\
        \hline
        Positive dependent & Has positively dependent  reaction vectors \\
        \hline
        Regular & Yes\\
        \hline
    \end{tabular}
    \caption{Some properties of AR system from CRNTToolBox}
    \label{table:arproperties}
\end{table}

From Section \ref{emission}, the kinetic order matrix $F$ of the power-law kinetics in the previous table is given below. Note that from Proposition \ref{opp point}, $e_{15}$ is the kinetic order of $A_4$ while $f_{15}$ is the kinetic order of $A_{15}$ such that $e_{15}\geq 1, f_{15}\leq 0,$ and $e_{15}+f_{15}=1$. 

\begin{equation*}
F= \begin{blockarray}{cccccl}
A_1 & A_2 & A_3 & A_4 & A_{15} & \\
\begin{block}{[ccccc]l}
p_1 & q_1 & 0 & 0 & 0 & R_1 \\
p_2 & q_2 & 0 & 0 & 0 & R_2 \\
0 & 1 & 0 & 0 & 0 & R_3 \\
0 & 0 & 1 & 0 & 0 & R_4 \\
0 & 0 & 0 & e_{15} & f_{15} & R_5 \\
1 & 0 & 0 & 0 & 0 & R_{15,6} \\
0 & 0 & 0 & 0 & 1 & R_{15,7} \\
\end{block}
\end{blockarray}
\end{equation*}

The stoichiometric matrix $N$ is also given below, and will result to the corresponding stoichiometric subspace $S$ which is the same with BECCS system. Note that the non-branching reactions of AR network imply that the system is PL-RDK with deficiency equal to one.
\begin{equation*}
N= \begin{blockarray}{cccccccl}
R_1 & R_{2} & R_3 & R_4 & R_5 & R_{15,6} & R_{15,7} &  \\
\begin{block}{[ccccccc]l}
1 & -1 & 0 & 0 & 0 & -1 & 0 & A_1 \\
-1 & 1 & 1 & -1 & 1 & 0 &  0 & A_2 \\
0 & 0 & -1 & 1 & 0 & 0 & 0 & A_3 \\
0 & 0 & 0 & 0 & -1 & 0 & 1 & A_4 \\
0 & 0 & 0 & 0 & 0 & -1 & -1 & A_{15} \\
\end{block}
\end{blockarray}
\end{equation*}

\[S=\text{span}\left\lbrace \begin{bmatrix} 1 \\ -1 \\ 0 \\ 0 \\ 0 \end{bmatrix}, \begin{bmatrix} 0 \\ -1 \\ 1 \\ 0 \\ 0 \end{bmatrix},\begin{bmatrix} -1 \\ 0 \\ 0 \\ 0 \\ 1 \end{bmatrix}, \begin{bmatrix} 0 \\ 0 \\ 0 \\ 1 \\ -1 \end{bmatrix}\right\rbrace\]

Hence, the ODE system associated with the power-law kinetics of the AR system is given below:

\begin{align*}
    \dot{A_1}=  &  k_1 A_1^{p_1} A_2^{q_1}-k_2 A_1^{p_2} A_2^{q_2}-k_{15,6}A_1\\
    \dot{A_2}=  &  k_2 A_1^{p_2} A_2^{q_2}-k_1 A_1^{p_1} A_2^{q_1}-a_m A_2+a_m\beta A_3 + k_5 A_4^{e_{15}}A_{15}^{f_{15}}\\
    \dot{A_3}=  &  a_m A_2-a_m\beta A_3\\
    \dot{A_4}= & k_{15,7} A_{15}-k_5 A_4^{e_{15}}A_{15}^{f_{15}} \\
    \dot{A_{15}}= & k_{15,6}A_1-k_{15,7} A_{15}.
\end{align*}

\subsection{Injectivity and multistationarity analysis}
To determine if an AR system is injective or not, we apply the method of Feliu and Wiuf \cite{WIUF2013} to compute for $M^*$ and determine the condition/s for kinetic orders that will make $M^*$ sign-non-singular.

\begin{proposition}
    On an AR system, the following classes are injective:
    \begin{itemize}
        \item [(i)] AR Negative with $p_1<0, p_2>0,q_1>0,q_2<0$;
        \item [(ii)] AR P-null with $p_1=p_2=0,q_1>0,q_2<0$; and
        \item [(iii)] AR Q-null with $p_1<0,p_2>0,q_1=q_2=0$
    \end{itemize}
    All other AR classes are non-injective.
\end{proposition}

\begin{proof}
    Using the computational approach and Maple script provided in \cite{FELIU2013}, the matrix $M^*$ is given below

\begin{equation*}
M^*= \begin{blockarray}{cccccl}
\begin{block}{[ccccc]l}
(z_1 p_1 z_2 p_2 z_6)k_1 & -(z_1 q_1 z_2 q_2)k_2 & 0 & 0 & 0 \\
(-z_2 p_1+z_2p_2)k_1 & (-z_1q_1-z_2q_2z_3)k_2 & z_4k_3 & z_5e_{4,15}k_4 & z_5f_{4,15}k_5 \\
0 & z_3k_2 & -z_4k_3 & 0 & 0 \\
0 & 0 & 0 & 0-z_5e_{4,15}k_4 & (-z_4f_{4,15}+z_7)k_5 \\
z_6k_1 & 0 & 0 & 0 & -z_7k_5 \\
\end{block}
\end{blockarray}
\end{equation*}

Thus, the determinant of $M^*$ is
\begin{align*}
det(M^*)&=e_{4,15}q_1k_1k_2k_3k_4z_1z_4z_5z_6-e_{4,15}q_2k_1k_2k_3k_4z_2z_4z_5z_6-e_{4,15}p_1k_1k_2k_4k_5z_2z_4z_5z_7
    \\&+e_{4,15}p_2k_1k_2k_4k_5z_2z_3z_5z_7-e_{4,15}p_1k_1k_3k_4k_5z_1z_4z_5z_7+e_{4,15}p_2k_1k_3k_4k_5z_2z_4z_5z_7
    \\&+e_{4,15}q_1k_2k_3k_4k_5z_1z_4z_5z_7-e_{4,15}q_2k_2k_3k_4k_5z_2z_4z_5z_7-f_{4,15}q_1k_1k_2k_3k_5z_1z_4z_5z_6    
 \\&+f_{4,15}q_2k_1k_2k_3k_5z_2z_4z_5z_6+e_{4,15}k_1k_2k_4k_5z_3z_5z_6z_7+e_{4,15}k_1k_3k_4k_5z_4z_5z_6z_7
    \\&+q_1k_1k_2k_3k_5z_1z_4z_5z_6-q_2k_1k_2k_3k_5z_2z_4z_6z_7
\end{align*}

Now, the terms $e_{4,15}k_1k_2k_4k_5z_3z_5z_6z_7\text{ and } e_{4,15}k_1k_3k_4k_5z_4z_5z_6z_7$ are always positive since $e_{4,15}\geq 0$. Hence, all the other terms of the determinant are all positive and $M^*$ is sign-non-singular which is injective, if
\begin{itemize}
    \item [(i)] $p_1<0,p_2>0,q_1>0,q_2<0$ or
    \item [(ii)] $p_1=0,p_2=0,q_1>0,q_2<0$ or
    \item [(iii)] $p_1<0,p_2>0,q_1=0,q_2=0$
\end{itemize}

Note that other combinations of signs of kinetic order will not make all the terms of the determinant positive. Thus, AR negative with $p_1<0,p_2>0,q_1>0,q_2<0$, P-null with $p_1=0,p_2=0,q_1>0,q_2<0$, and Q-null with $p_1<0,p_2>0,q_1=0,q_2=0$ are injective and does not have the capacity for multiple steady states.
\end{proof}

To establish the capacity of multiple steady states of the AR system, we once again apply the Deficiency-One Algorithm. The proof of the following result is provided in Appendix \ref{appendix:AR}.

\begin{proposition} \label{prop:DOA AR}
    A positive AR system with kinetic orders $p_1>p_2>1$ and $q_1>q_2>0$ is multistationary. 
\end{proposition}

\subsection{ACR Analysis of AR system}
The method of determining species of AR system that admits ACR will be the same as computing in that of BECCS. The parameterization of $A_1$ will be used to determine the symbolic properties of other species in a steady state.
\\
\begin{proposition}\label{prop:ACRAR}
    A P-null AR system with $p_1=p_2=1$ admits ACR only on species $A_2$ and $A_3$. All other AR subsets will not admit ACR on any species.
\end{proposition}
 \begin{proof}
     To ease the computation of steady states, we use the ODE system of AR network's finest independent decomposition and use the partial equilibria parametrization. The finest independent decomposition of AR is as follows:
     \[\mathscr{P}_1=\{A_1+2A_2\leftrightarrows 2A_1+A_2,A_1\to A_{15}\to A_4, A_4+A_{15}\to A_2+A_{15}\}, \mathscr{P}_2=\{A_2\leftrightarrows A_3\}.\]
The corresponding ODE system of $\mathscr{P}_1$ is given by:
     \begin{align*}
    \dot{A_1}=  &  k_1 A_1^{p_1} A_2^{q_1}-k_2 A_1^{p_2} A_2^{q_2}-k_{15,6}A_1\\
    \dot{A_2}=  &  k_2 A_1^{p_2} A_2^{q_2} -k_1 A_1^{p_1} A_2^{q_1}+ k_5 A_4^{e_{15}}A_{15}^{f_{15}}\\
    \dot{A_4}= & k_{15,7} A_{15}-k_5 A_4^{e_{15}}A_{15}^{f_{15}}\\
    \dot{A_{15}}= & k_{15,6}A_1-k_{15,7} A_{15}
\end{align*}
On the other hand, the ODE system of $\mathscr{P}_2$ is given by:
\begin{align*}
    \dot{A_2}=  &  -a_m A_2+a_m\beta A_3\\
    \dot{A_3}=  &  a_m A_2-a_m\beta A_3
\end{align*}
Using the ODE system of $\mathscr{P}_1$, set each equation to zero to obtain the steady states and let $A_1=\tau_1>0$. We can parametrize the values of $A_4$ and $A_{15}$ as follows: 
\[A_{15}=\frac{k_{15,6}}{k_{15,7}}\tau_1, A_4^{e_{15}}=\frac{k_{15,7} A_{15}}{k_5 A_{15}^{f_{15}}}=\frac{k_{15,7}}{k_5}\left(\frac{k_{15,6}}{k_{15,7}}\tau_1\right)^{e_{15}}\Longrightarrow A_4=\left(\frac{k_{15,7}}{k_5}\right)^{\frac{1}{e_{15}}}\frac{k_{15,6}}{k_{15,7}}\tau_1\]
Thus, $A_2= Root~\{k_2 \tau_1^{p_2} A_2^{q_2} -k_1 \tau_1^{p_1} A_2^{q_1}+ k_{15,6}\tau_1 \}$.\\

Using the ODE system of $\mathscr{P}_2$, let $A_2=\tau_2>0$ so that $A_3=\frac{1}{\beta}\tau_2$. Note that $A_{15}$ and $A_4$ are both dependent on $\tau_1$. The parameterization of $A_1$ means that it cannot admit an ACR, this is the same for $A_{15}$ and $A_4$.\\

Since $A_2= Root~\{k_2 \tau_1^{p_2} A_2^{q_2} -k_1 \tau_1^{p_1} A_2^{q_1}+ k_{15,6}\tau_1 \}$, we observe that
\begin{align*}
&k_2 \tau_1^{p_2} A_2^{q_2} -k_1 \tau_1^{p_1} A_2^{q_1}+ k_{15,6}\tau_1=0,\\ 
\Rightarrow & k_2 \tau_1^{p_2-1} A_2^{q_2} -k_1 \tau_1^{p_1-1} A_2^{q_1}=k_{15,6},\\ 
\Rightarrow &\tau_1^{p_2-1}(k_1\tau_1^{p_2-p_1}A_2^{q_1}-k_2 A_2^{q_2})=k_{15,6}.
\end{align*}
Using $p_2=p_1=1$, the LHS is $k_1A_2^{q_1}-k_2 A_2^{q_2}$ which is dependent on $A_2$ and the RHS is a constant. And so, $A_2$ is invariant on any steady state if $p_2=p_1=1$ and admits ACR. Since $A_3$ is dependent on $A_2$, it also admits ACR. Note that other values of the kinetic orders will result to different values of each species in steady states. Hence, any other subsets of AR system will not admit ACR on any species.
 \end{proof}

The analyses for AR system without natural reforestation are summarized in the table below.

\begin{table}[h]
    \centering
\begin{adjustbox}{width=1\textwidth}
\small
    \begin{tabular}{|l|l|l|}
    \hline
        \textbf{AR class} & \textbf{Mono-/multi-stationary subsets} & \textbf{Injective/non-injective subsets} \\
    \hline
        AR positive &  Multistationary: $p_1>p_2>1$ and $q_1>q_2>0$  &  Non-injective for all values of kinetic orders \\
    \hline
        AR negative &  Monostationary: $p_1<0$, $p_2>0$, $q_1>0$, and $q_2<0$ & Injective: $p_1<0$, $p_2>0$, $q_1>0$, and $q_2<0$  \\
    \hline
        AR $P$-null & Monostationary: $p_1=p_2=0$, $q_1>0$, and $q_2<0$  &  Injective: $p_1=p_2=0$, $q_1>0$, and $q_2<0$ \\
    \hline
        AR $Q$-null & Monostationary: $q_1=q_2=0$, $p_1<0$, and $p_2>0$ & Injective: $q_1=q_2=0$, $p_1<0$, and $p_2>0$ \\
    \hline
\end{tabular}
\end{adjustbox}
    \caption{Summary of mutistationarity and injectivity analysis for BECCS }
    \label{table: AR summary}
\end{table}

\section{Comparison of BECCS and AR Kinetic Systems}

The differences between the BECCS and AR systems arise from the reaction $R_5$, which represents the carbon emission from TCS into the atmosphere. As discussed in Section \ref{kinetics}, while neither the BECCS nor the AR systems store inorganic carbon output in TCS (i.e., $\mu_i=0$), they differ in the amounts of carbon sequestered by their CDR storage that remains in TCS as long-term carbon (indicated by the parameter $\lambda_i$). The difference in the corresponding rate functions in $R_5$ led to different reaction representations in their corresponding CRN models.

However, despite this particular difference in reactions, there seem to be no notable variations in the graphical or dynamic properties of the two networks. The differences in their kinetic representations are only apparent in network numbers, such as the number of complexes, linkage classes, strong linkage classes, and terminal strong linkage classes. All other structural properties are largely similar. Furthermore, since the reaction vectors of both systems are the same, their stoichiometry-related properties also align. Table \ref{table:comparison} summarizes the different graphical, stoichiometric, and dynamical characteristics of the BECCS and AR systems.

\begin{table}[!ht]
\centering
\begin{minipage}{\textwidth}
\begin{adjustbox}{width=1\textwidth}
\small
\begin{tabular}{|c|l|l|}
\hline
\textbf{Property class} & \textbf{BECCS system} & \textbf{AR sytem} \\ 
\hline
\multirow{3}{*}{Network}    &   Disconnected with 2 linkage classes &   Disconnected with 4 linkage classes \\
\hhline{~--}  & Non-weakly reversible & Non-weakly reversible \\
\hhline{~--} (Digraph)& Non-cycle terminal &    Non-cycle terminal \\
\hhline{~--} &  $t$-minimal &   $t$-minimal \\
\hline
\multirow{8}{*}{Network} & Discordant    & Discordant    \\
\hhline{~--} & Conservative & Conservative  \\
\hhline{~--} & Positive dependent & Positive dependent \\
\hhline{~--} &  Deficiency-one & Deficiency-one\\
\hhline{~--} (Stochiometry-related) &  Regular & Regular \\
\hhline{~--} &  FID has 2 subnetworks,  & FID has 2 subnetworks \\
\hhline{~~~} &  neither of which is an Anderies subsystem & neither of which is an Anderies subsystem \\
\hline
\multirow{1}{*}{ Kinetic system}  & PL-RDK &  PL-RDK  \\
\hline
\multirow{2}{*}{ACR}  & No ACR species except  &  No ACR species except  \\
\hhline{~~~} & when $p_1=p_2=1$ (Prop. \ref{prop:ACRBECCS}) & when $p_1=p_2=1$ (Prop. \ref{prop:ACRAR})  \\
\hline
\multirow{4}{*}{Injectivity}  & Positive:  Non-injective &  Positive:  Non-injective    \\
\hhline{~--}  & Negative: Has injective subset &  Negative: Has injective subset  \\
\hhline{~--}  & $P$-null: Has injective subset &  $P$-null: Has injective subset  \\
\hhline{~--}  & $Q$-null: Has injective subset &  $Q$-null: Has injective subset  \\
\hline
\multirow{4}{*}{Equilibria multiplicity }  & Positive:  contains multistationary systems &  Positive: contains multistationary systems    \\
\hhline{~--}  & Negative: contains monostationary systems &  Negative: contains monostationary systems  \\
\hhline{~--}  & $P$-null: contains monostationary systems &  $P$-null: contains monostationary systems  \\
\hhline{~--}  & $Q$-null: contains monostationary systems &  $Q$-null: contains monostationary systems  \\
\hline
\end{tabular}
\end{adjustbox}
\caption{Overview of network and kinetic properties of the BECCS and AR systems.}
\label{table:comparison}
\end{minipage}
\end{table}  

\newpage

\section{Summary and Outlook}
In this study, we proposed a framework based on chemical reaction network theory to analyze several negative emissions technology models to capture or remove carbon dioxide in present-day Earth. Using this framework, we looked into the conditions as to when these systems would exhibit steady-state multiplicity and absolute concentration robustness.

The proposed framework is called Reaction Network Carbon Dioxide Removal (RNDCR) highlighted by two components common to all models: the Anderies pre-industrial subsystem and the fossil fuel emission reaction. We formed the RNDCR network consisting of its species, complexes, reactions and corresponding power-law kinetics. The classes and Anderies decomposition of an RNDCR system were also discussed.

We illustrated the RNDCR framework in the context of Bioenergy with Carbon Capture and Storage (BECCS) and afforestation/reforestation (AR) which are considered two of the mature NETs. 

These two NET models differ only in one reaction, that is, the reaction representing the carbon emission from the total carbon stock into the atmosphere. Their corresponding CRN representations differ due to the rate functions associated with this reaction. We showed that despite this difference, the graphical and dynamical properties of the two CRN models are almost similar.

The ACR analysis of BECCS and AR revealed that there is a particular condition under which a P-null system will demonstrate ACR in both species $A_2$ (atmosphere) and $A_3$ (ocean). All other BECCS/AR subsets will not admit ACR on any species.

On the other hand, it was discovered through the Deficiency One Algorithm (DOA) that there exist appropriate kinetic order values such that BECCS or AR systems are multistationary, that is, they have the capacity to admit multiple steady states, and thus "tipping points".

For future studies, the analysis of other NETs can be done using the RNDCR framework. 

\section*{Declarations}

\begin{itemize}
\item Funding- Not applicable
\item Conflict of interest/Competing interests-The authors declare no competing interests.
\item Ethics approval and consent to participate- Not applicable
\item Consent for publication- Not applicable
\item Data availability - Not applicable
\item Materials availability- Not applicable
\item Code availability- Not applicable
\end{itemize}



\begin{thebibliography}{00}

\bibitem{ANDE2013}
 J.M. Anderies, S.R. Carpenter, W. Steffen, J. Rockström, The topology of non-linear global carbon dynamics: from tipping points to planetary boundaries, \textit{Environ. Res. Lett.} \textbf{8} (4), 044--048 (2013)

\bibitem{ABDFHR2023}
J. M. Anderies, W.Barfuss, J. F. Donges, I. Fetzer,  J. Heitzig, J. Rockström, A modeling framework for World-Earth system resilience: exploring social inequality and Earth system tipping points, \textit{Environ. Res. Lett.} \textbf{18} (9),  (2023)

\bibitem{AJMSM2015}
C. P. Arceo, E. Jose, A. Mar{\'i}n-Sanguino, E. Mendoza, Chemical reaction network approaches to Biochemical Systems Theory, \textit{Math. Biosci.} \textbf{269}, 135--152 (2015)

\bibitem{AJLM2017}
 C.P.P. Arceo, E. C. Jose, A.R. Lao, E. R. Mendoza, Reaction networks and kinetics of biochemical systems,
 \textit{Math. Biosci.} \textbf{283}, 13--29 (2017)


\bibitem{FEIN1987} 
M. Feinberg, Chemical reaction network structure and the stability of complex isothermal reactors {I}: The deficiency zero and deficiency one theorems, \textit{Chem. Eng. Sci.} \textbf{42} (10), 2229--2268 (1987)

\bibitem{FEIN1995DOA}
M. Feinberg, Multiple Steady States for Chemical Reaction Networks of Deficiency One, \textit{Arch. Ration. Mech. Anal} \textbf{132}, 371--406 (1995)

\bibitem{CRNTool}
M. Feinberg, P. Ellison, H. Ji, D. Knight, The Chemical Reaction Network Toolbox Version 2.35, (2018)
doi.org/10.5281/zenodo.5149266

\bibitem{FEIN2019}
 M. Feinberg, Foundations of Chemical Reaction Network Theory, \textit{Springer International Publishing}, Switzerland (2019)


\bibitem{FELIU2013}
E. Feliu, C. Wiuf, A computational method to preclude multistationarity in networks of interacting species, \textit{Bioinformatics} \textbf{29} (18), 2327--2334 (2013)


\bibitem{FMRL2018}
N.T. Fortun, E.R. Mendoza, L.F. Razon, A. R. Lao, A Deficiency One Algorithm for Power in Law Kinetic Systems with Reactant-Determined Interactions, \textit{J. Math. Chem.}, \textbf{56} (10), 2929--2962 (2018)

\bibitem{MJS2018}
N. Fortun, A. Lao, L. Razon, E. Mendoza, Multistationarity in Earth's pre-industrial carbon cycle models, \textit{Manila J. Sci} \textbf{11}, 81--96 (2018)

\bibitem{FMRL2019}
 N.T. Fortun, E.R. Mendoza, L.F. Razon, A. R. Lao, A deficiency zero theorem for a class of power–law kinetic systems with non–reactant–determined interactions, \textit{MATCH Commun. Math. Comput. Chem.} \textbf{81} (3), 621--638 (2019)

\bibitem{FOME2021}
N.T. Fortun, E.R. Mendoza, Absolute concentration robustness in power law kinetic systems, \textit{MATCH Commun. Math. Comput. Chem.}  \textbf{85} (3), 669--691 (2021)
 
\bibitem{FOME2023}
 N.T. Fortun, E.R. Mendoza, Comparative analysis of carbon cycle models via kinetic representations,
  \textit{J. Math. Chem.} \textbf{61} (5), 896--932 (2023)

\bibitem{FLMR2024}
N.T. Fortun, A. R. Lao, E.R. Mendoza, L.F. Razon, Determining the Possibility of Multistationarity in a Model of the Earth Carbon Cycle with Direct Air Capture, \textit{submitted}, Preprint at https://arxiv.org/abs/2405.17058

\bibitem{HEITZIG2016}
J. Heitzig, T. Kittel, J. F. Donges,  N. Molkenthin,  Topology of sustainable management of dynamical systems with desirable states: From defining planetary boundaries to safe operating spaces in the Earth system, \textit{Earth Syst. Dyn.} \textbf{7} (1), 21--50 (2016)

\bibitem{JOHNSTON2014}
M.D. Johnston, Translated chemical reaction networks, \textit{Bull. Math. Biol.} \textbf{76} (5), 1081--1116 (2014)
  
\bibitem{LLMM2022}
 A.R. Lao, P.V.N. Lubenia, D.M. Magpantay, E.R. Mendoza, Concentration robustness in LP kinetic systems,
\textit{MATCH Commun. Math. Comput. Chem.}  \textbf{88} (1), 29--66 (2022)

\bibitem{LENTON2000}
T. M.  Lenton, Land and ocean carbon cycle feedback effects on global warming in a simple Earth system model, \textit{TELLUS B} \textbf{52} (5), 1159--1188 (2000)

\bibitem{MCKAY2022}
D. I. A.  McKay, A. Staal, A., J. F. Abrams, R. Winkelmann, B. Sakschewski, S. Loriani, I. Fetzer, S. E. Cornell, J.  Rockström, T.M. Lenton,  Exceeding 1.5°C global warming could trigger multiple climate tipping points, \textit{Science} \textbf{377}, (2022)

\bibitem{CB2018}
R. Mcsweeney,  Z. Hausfather, Q and A: How do climate models work?, (2018), https://www.carbonbrief.org/qa-how-do-climate-models-work/\#what"

\bibitem{MURE2012}
S. M{\"u}ller, G. Regensburger, Generalized mass action systems: Complex balancing equilibriaand sign vectors of the stoichiometric and kinetic-order subspaces, \textit{SIAM J. Appl. Math.} \textbf{72} (6), 1926--1947 (2012)

\bibitem{MURE2014}
S. M{\"u}ller, G. Regensburger, Generalized mass-action systems and positive solutions of polynomial equations with real and symbolic exponents (Invited Talk), In: Gerdt, V., Koepf, W., Seiler, W., Vorozhtsov, E. (eds.) Computer Algebra in Scientific Computing, 302--323, Springer, Cham (2014)

\bibitem{NITZBON2017}
 J. Nitzbon, J. Heitzig, U. Parlitz, Sustainability, collapse and oscillations in a simple World-Earth model,
\textit{Environ. Res. Lett.} \textbf{12} (7), 074020 (2017)
 
\bibitem{RAZON1987}
L. F. Razón, R. A. Schmitz, Multiplicities and instabilities in chemically reacting systems - a review, 
\textit{Chem. Eng. Sci.} \textbf{42} (5), 1005--1047 (1987)
  
\bibitem{SCHMITZ2002}
R.A.Schmitz, The Earth's Carbon Cycle: Chemical Engineering Course Material, \textit{Chemical Engineering Education}, 296--303, 309 (2002)

 \bibitem{SHFE2012}
G. Shinar, M. Feinberg, Concordant chemical reaction networks, \textit{Math. Biosci.} \textbf{240} (2), 92--113 (2012)
 
\bibitem{SHFE2010}
G. Shinar, M. Feinberg, Structural sources of robustness in biochemical reaction networks, \textit{Science} \textbf{327} (5971), 1389--1391 (2010)

\bibitem{SHUKLA2022}
P. R. Shukla, et. al., Climate Change 2022 Mitigation of Climate Change Working Group III Contribution to the Sixth Assessment Report of the Intergovernmental Panel on Climate Change Summary for Policymakers, (2022), Edited by www.ipcc.ch


\bibitem{STEF2015}
W. Steffen and et.al., Planetary boundaries: Guiding human development on a changing planet, \textit{Science} 
\textbf{347} (6223), (2015)

\bibitem{STREF021}
J. Strefler, N. Bauer, F. Humpen{\"o}der, D. Klein, A. Popp, E. Kriegler, Carbon dioxide removal technologies are not born equal, \textit{Environ. Res. Lett.} \textbf{16} (7), 074021 (2021)


\bibitem{TALABIS2017}
  D.A. Talabis, C.P. Arceo, E. Mendoza, Positive equilibria of a class of power-law kinetics, \textit{J. Math. Chem.} \textbf{56} (2), 358--394 (2017)

\bibitem{TAN2022}		
R. R. Tan, K. B. Aviso,  D. C. Y. Foo, M. V. Migo-Sumagang, P. N. S. B. Nair, M. Short , Computing optimal carbon dioxide removal portfolios, \textit{Nat. Comput. Sci.} \textbf{2} (8), 465--466 (2022)

\bibitem{WIUF2013}
 C. Wiuf, E. Feliu, Power-Law Kinetics and Determinant Criteria for the Preclusion of Multistationarity in Networks of Interacting Species, \textit{SIAM J. Appl. Dyn. Syst.} \textbf{12} (4), 1685--1721 (2013)


\end{thebibliography}

\newpage
\begin{appendix}

\section{Fundamentals of reaction networks and kinetic systems }\label{appendix: CRNT}

As supplementary material, we provide a formal presentation of the relevant concepts and results related to chemical reaction networks and chemical kinetic systems.

\subsubsection*{Notation}
We denote the real numbers by $\mathbb{R}$, the non-negative real numbers by $\mathbb{R}_{\geq0}$ and the positive real numbers by $\mathbb{R}_{>0}$.  Objects in reaction systems are viewed as members of vector spaces. Suppose $\mathscr{I}$ is a finite index set. By $\mathbb{R}^\mathscr{I}$, we mean the usual vector space of real-valued functions with domain $\mathscr{I}$.  If $x \in \mathbb{R}_{>0}^\mathscr{I}$ and $y \in \mathbb{R}^\mathscr{I}$, we define $x^y \in \mathbb{R}_{>0}$ by
$
x^y= \prod_{i \in \mathscr{I}} x_i^{y_i} .
$
Let $x \wedge y$ be the component-wise minimum, $(x \wedge y)_i = \min (x_i, y_i)$.
The vector $\log x\in \mathbb{R}^\mathscr{I}$,where $x \in \mathbb{R}_{>0}^\mathscr{I}$, is given by 
$(\log x)_i = \log x_i,  \text{ for all } i \in \mathscr{I}.$  The \textit{support} of $x \in \mathbb{R}^\mathscr{I}$, denoted by $\text{supp } x$, is given by
$ \text{supp } x := \{ i \in \mathscr{I} \mid x_i \neq 0 \}.$

\subsection{Fundamentals of chemical reaction networks}

We begin with the formal definition of a chemical reaction network or CRN. 

\begin{definition}
A \textbf{chemical reaction network} or CRN is a triple $\mathscr{N}:= (\mathscr{S,C,R})$ of nonempty finite sets $\mathscr{S}$, $\mathscr{C}$, and $\mathscr{R}$, of $m$ \textbf{species}, $n$ \textbf{complexes}, and $r$ \textbf{reactions}, respectively, where $\mathscr{C} \subseteq \mathbb{R}_{\geq 0}^\mathscr{S}$ and $\mathscr{R} \subset \mathscr{C} \times \mathscr{C}$ satisfying the following properties:
\begin{enumerate}
    \item $(y,y) \notin \mathscr{R}$ for any $y \in \mathscr{C}$;
    \item for each $y \in \mathscr{C}$, there exists $y' \in \mathscr{C}$ such that $(y,y')\in \mathscr{R}$ or $(y',y)\in \mathscr{R}$.
\end{enumerate}
\end{definition}
\noindent For $y \in \mathscr{C}$, the vector $$y=\displaystyle{\sum_{A \in \mathscr{S}}} y_A A,$$ where $y_A$ is the \textbf{stoichiometric coefficient} of the species $A$. In lieu of $(y,y')\in \mathscr{R}$, we write the more suggestive notation  $y \rightarrow y'$. In this reaction, the vector $y$ is called the \textbf{reactant complex} and $y'$ is called the \textbf{product complex}. 

CRNs can be viewed as directed graphs where the complexes are vertices and the reactions are arcs. The (strongly) connected components are precisely the \textbf{(strong) linkage classes} of the CRN. A strong linkage class is a \textbf{terminal strong linkage class} if there is no reaction from a complex in the strong linkage class to a complex outside the given strong linkage class. 

\begin{definition}
A CRN with $n$ complexes, $n_r$ reactant complexes, $\ell$ linkage classes, $s\ell$ strong linkage classes, and $t$ terminal strong linkage classes is 
\begin{enumerate}
    \item \textbf{weakly reversible} if $s\ell = \ell$;
    \item \textbf{t-minimal} if $t=\ell$;
    \item \textbf{point terminal} if $t=n-n_r$; and
    \item \textbf{cycle terminal} if $n-n_r=0$.
\end{enumerate}
\end{definition}

For every reaction, we associate a \textbf{reaction vector}, which is obtained by subtracting the reactant complex from the product complex. From a dynamic perspective, every reaction $ y \rightarrow y' \in \mathscr{R}$ leads to a change in species concentrations proportional to the  reaction vector $ \left( y' – y \right) \in \mathbb{R}^\mathscr{S}$. The overall change induced by all the reactions lies in a subspace of $\mathbb{R}^\mathscr{S}$ such that any trajectory in $\mathbb{R}^\mathscr{S}_{>0}$ lies in a coset of this subspace. 

\begin{definition}
The \textbf{stoichiometric subspace} of a network $\mathscr{N}$ is given by
$$ \mathcal{S} := \text{span } \{ y' – y \in \mathbb{R}^\mathscr{S} \mid y \rightarrow y' \in \mathscr{R} \}.$$
The \textbf{rank} of the network is defined as $s:= \dim \mathcal{S}$. For $x \in \mathbb{R}^\mathscr{S}_{>0}$, its \textbf{stoichiometric compatibility class} is defined as $(x+\mathcal{S}) \cap \mathbb{R}^\mathscr{S}_{ \geq 0}$. Two vectors $x^{*}, x^{**} \in  \mathbb{R}^\mathscr{S}$ are \textbf{stoichiometrically compatible} if $ x^{**}-x^{*} \in \mathcal{S}$.
\end{definition}

\begin{definition}
A CRN with stoichiometric subspace $S$ is said to be \textbf{conservative} if there exists a positive vector $x \in \mathbb{R}^\mathscr{S}_>$ such that $S^\perp \cap \mathbb{R}^\mathscr{S}_> \neq \emptyset$. 
\end{definition}

An important structural index of a CRN, called \textit{deficiency}, provides one way to classify networks.

\begin{definition}
The \textbf{deficiency} $\delta$ of a CRN with $n$ complexes, $\ell$ linkage classes, and rank $s$ is defined as $\delta:=n-\ell-s$.
\end{definition}

\subsection{Fundamentals of chemical kinetic systems}
It is generally assumed that the rate of a reaction $y \rightarrow y' \in \mathscr{R}$ depends on the concentrations of the species in the reaction. The exact form of the non-negative real-valued rate function $K_{ y \rightarrow y'}$ depends on the underlying \textit{kinetics}.

The following definition of kinetics is expressed in a more general context than what one typically finds in CRNT literature.
\begin{definition}
A \textbf{kinetics} for a network $\mathscr{N}=(\mathscr{S,C,R})$ is an assignment to each reaction $y \rightarrow y' \in \mathscr{R}$ a rate function $ K_{ y \rightarrow y'}: \Omega_K \rightarrow \mathbb{R}_{\geq 0}$, where $\Omega_K$ is a set such that $\mathbb{R}^\mathscr{S}_{> 0} \subseteq \Omega_K \subseteq \mathbb{R}^\mathscr{S}_{\geq 0}$, $x  \wedge  x^{*} \in \Omega_K$ whenever $x, x^{*} \in \Omega_K$, and $ K_{ y \rightarrow y'} (x) \geq 0$ for all $x \in \Omega_K$. A kinetics for a network $\mathscr{N}$ is denoted by $K:\Omega_K \rightarrow \mathbb{R}^\mathscr{R}_{\geq 0}$ (\cite{WIUF2013}). A \textbf{chemical kinetics} is a kinetics $K$ satisfying the condition that for each $y \rightarrow y' \in \mathscr{R}$, $ K_{ y \rightarrow y'} (x) >0$ if and only if $\text{supp } y \subset \text{supp } x$. The pair $(\mathscr{N},K)$ is called a \textbf{chemical kinetic system} (\cite{AJLM2017}). 
\end{definition}

The system of ordinary differential equations that govern the dynamics of a CRN is defined as follows.

\begin{definition}\label{def:ODE}
The \textbf{ordinary differential equation (ODE)} associated with a chemical kinetic system $(\mathscr{N},K)$ is defined as 
$ \dfrac{dx}{dt}=f(x)$ with \textbf{species formation rate function} 
\begin{equation}\label{eq:sfrf}
    f(x)= \sum_{ y \rightarrow y' \in \mathscr{R}} K_{ y \rightarrow y'} (x) (y'-y).
\end{equation}
A \textbf{positive equilibrium} or \textbf{steady state} $x$ is an element of $\mathbb{R}^\mathscr{S}_{>0}$ for which $f(x) = 0$.
\end{definition}

\begin{definition}\label{def:equilibria}
The \textbf{set of positive equilibria} or \textbf{steady states} of a chemical kinetic system $(\mathscr{N},K)$ is given by 
$$ E_+ (\mathscr{N},K) = \{ x \in \mathbb{R}^\mathscr{S}_{>0} \mid f(x) = 0 \}. $$
For brevity, we also denote this set by $E_+$. The chemical kinetic system is said to be \textbf{multistationary} (or has the capacity to admit \textbf{multiple steady states}) if there exist positive rate constants such that $\mid E_+ \cap \mathcal{P}\mid \geq 2$ for some positive stoichiometric compatibility class $\mathcal{P}$. On the other hand, it is \textbf{monostationary} if $\mid E_+ \cap \mathcal{P}\mid \leq 1$ for all positive stoichiometric compatibility class $\mathcal{P}$.
\end{definition}

\begin{definition}
The reaction vectors of a CRN $ (\mathscr{S,C,R})$ are \textbf{positively dependent} if for each reaction $y \rightarrow y' \in \mathscr{R}$, there exists a positive number $k_{ y \rightarrow y'}$ such that $\sum_{y \rightarrow y' \in \mathscr{R}}k_{ y \rightarrow y'} (y'-y)=0$. 
\end{definition}

\begin{remark}
 In view of Definitions \ref{def:ODE} and  \ref{def:equilibria}, a necessary condition for a chemical kinetic system to admit a positive steady state is that its reaction vectors are positively dependent.  
\end{remark}

To reformulate the species formation rate function in Eq. (\ref{eq:sfrf}), we consider the natural basis vectors $\omega_i \in \mathbb{R}^\mathscr{I}$ where $i \in \mathscr{I}=\mathscr{C}$ or $\mathscr{R}$ and define 
\begin{enumerate}
    \item[(i)] the \textbf{molecularity map} $Y: \mathbb{R}^\mathscr{C} \rightarrow \mathbb{R}^\mathscr{S}$ with $Y(\omega_y)=y$;
    \item[(ii)] the \textbf{incidence map} $I_a: \mathbb{R}^\mathscr{R} \rightarrow \mathbb{R}^\mathscr{S}$ with $I_a (\omega_{y \rightarrow y'})= \omega_{y'} - \omega_y$; and
    \item[(iii)] the \textbf{stoichiometric map} $N: \mathbb{R}^\mathscr{R} \rightarrow \mathbb{R}^\mathscr{S}$ with $N= YI_a$.
\end{enumerate}
Hence, Eq. (\ref{eq:sfrf}) can be rewritten as $f(x)=YI_aK(x)=NK(x).$ The positive steady states of a chemical kinetic system that satisfies $I_a K(x)=0$ are called \textit{complex balancing equlibria}. 
\begin{definition}
The \textbf{set of complex balanced equilibria} of a chemical kinetic system $(\mathscr{N},K)$ is the set
$$ Z_+(\mathscr{N},K) = \{ x\in \mathbb{R}^\mathscr{S}_{>0} \mid I_a K(x) =0 \} \subseteq E_+(\mathscr{N},K).$$
A chemical kinetic system is said to be \textbf{complex balanced} if it has a complex balanced equilibrium. 
\end{definition}

We define power law kinetics through the $r \times m$ \textbf{kinetic order matrix} $F=[F_{ij}]$, where $F_{ij} \in \mathbb{R}$ encodes the kinetic order the $j$th species of the reactant complex in the $i$th reaction. Further, consider the \textbf{rate vector} $k \in \mathbb{R}^\mathscr{R}_{>0}$, where $k_i \in \mathbb{R}_{>0}$ is the rate constant in the $i$th reaction. 

\begin{definition}\label{def:PLK}
A kinetics $K: \mathbb{R}^\mathscr{S}_{>0} \rightarrow \mathbb{R}^\mathscr{R}$ is a \textbf{power law kinetics} or \textbf{PLK} if
$$\displaystyle K_{i}(x)=k_i x^{F_{i,*}} \quad \text{for all } i \in \mathscr{R},$$
where $F_{i,*}$ is the row vector containing the kinetic orders of the species of the reactant complex in the $i$th reaction.
\end{definition}
 

\begin{definition} A PLK system has \textbf{reactant-determined kinetics} (or of type \textbf{PL-RDK}) if for any two \textbf{branching reactions} $i$, $j \in \mathscr{R}$ (i.e., reactions sharing a common reactant complex), the corresponding rows of kinetic orders in $F$ are identical. That is, $F_{ih}=F_{jh}$ for all $h  \in \mathscr{S}$.  
\end{definition}


In the implementation of the DOA in Appendix \ref{appendix:DOA}, kinetic orders are encoded using $T$-matrix  \cite{TALABIS2017}. 
This matrix is derived from the $m \times n$ matrix $\widetilde{Y}$ defined by M\"{u}ller and Regensburger in \cite{MURE2012}. In this matrix,  $( \widetilde{Y})_{ij} = F_{ki}$ if $j$ is a reactant complex of reaction $k$ and $( \widetilde{Y})_{ij} = 0$, otherwise. 

\begin{definition}
The $m \times n_r$ $\bm{T}$\textbf{-matrix} is the truncated $\widetilde{Y}$ where the non-reactant columns are deleted and $n_r$ is the number of reactant complexes. 
\end{definition}



\subsection{Independent decomposition of a CRN}
Decomposition theory was initiated by M. Feinberg in his 1987 review paper \cite{FEIN1987}. He introduced the general concept of a network decomposition of a CRN as a union of subnetworks whose reaction sets form a partition of the network’s set of reactions. He also introduced the so-called \textit{independent decomposition} of chemical reaction networks.

\begin{definition}
A decomposition of a CRN $\mathscr{N}$ into $k$ subnetworks of the form $\mathscr{N}=\mathscr{N}_1 \cup \cdots \cup \mathscr{N}_k$ is \textbf{independent} if its stoichiometric subspace is equal to the direct sum of the stoichiometric subspaces of its subnetworks, i.e., $\mathcal{S}=\mathcal{S}_1 \oplus \cdots \oplus \mathcal{S}_k$.
\end{definition}

For an independent decomposition, Feinberg concluded that any positive equilibrium of the “parent network” is also a positive equilibrium of each subnetwork.

\begin{theorem}[Rem. 5.4, \cite{FEIN1987}]  \label{feinberg theorem}
Let $(\mathscr{N},K)$ be a chemical kinetic system with partition $\{\mathscr{R}_1, \dots, \mathscr{R}_k \}$. If $\mathscr{N}=\mathscr{N}_1 \cup \cdots \cup\mathscr{N}_k$ is the network decomposition generated by the partition  and $E_+(\mathscr{N}_i,K_i)= \{ x \in \mathbb{R}^\mathscr{S}_{>0} \mid N_i K_i(x) = 0, i=1,\dots,k \}$, then 
$ \bigcap_{i=1}^kE_+ (\mathscr{N}_i, K_i)\subseteq E_+ (\mathscr{N}, K)$. If the network decomposition is independent, then equality holds.
\end{theorem}

\subsection{Absolute concentration robustness in PLP systems}\label{appendix:PLP}

Lao et al. \cite{LLMM2022} introduced the concept of positive equilibria log-parametrized (PLP) kinetic system:

\begin{definition}
For a reaction network $\mathscr{N}$ with species $\mathscr{S}$, a \textbf{log-parametrized (LP) set} is a non-empty set of the form $$E(P,x^*)= \{ x\in \mathbb{R}^\mathscr{S}_{>0} \mid \log x -\log x^* \in P^\perp \},$$ where $P$ (called the LP set's \textbf{flux subspace}) is a subspace of $\mathbb{R}^\mathscr{S}$, $x^*$ (called the LP's \textbf{reference point}) is a given element of $\mathbb{R}^\mathscr{S}_{>0}$, and $P^\perp$ (called the LP set's \textbf{parameter subspace}) is the orthogonal complement of $P$. A chemical kinetic system $(\mathscr{N},K)$ is \textbf{positive equilibria log-parametrized (PLP) system} if its set of positive equilibria is an LP set, i.e., $E_+(\mathscr{N},K)=E(P_E,x^*)$ where $P_E$ is the flux subspace and $x^*$ is a given positive equilibrium.
\end{definition}

The \textit{Species Hyperplane Criterion} for absolute concentration robustness is recalled below.

\begin{theorem}[Theorem 3.12, \cite{LLMM2022}] 
If $(\mathscr{N},K)$ is a PLP system, then it has ACR is a species $A$ if and its parameter subspace $P_E^\perp$ is a subspace of the species hyperplane $\{x \in  \mathbb{R}^\mathscr{S} \mid x_A =0 \}$.
\end{theorem}

\begin{remark}
The flux subspace of the Anderies system is its \textbf{kinetic flux subspace}, denoted by $\widetilde{S}$. This subspace is the kinetic analogue of the stoichiometric subspace. If the stoichiometric subspace is the span of the reaction vectors, the kinetic flux subspace is the span of the fluxes in terms of the kinetic vectors. In light of the discussion above, if the vector $x^*$ is any positive steady state of the system, the set of positive equilibria consists of vectors $x$ such that the vector $\log (x) - \log (x^*)$ resides in $\widetilde{S}^\perp$. Specifically, for the Anderies system, $\widetilde{S}^\perp = \text{span } \{ [-Q, 1, 1]^\top \}$ where $Q = \dfrac{q_2- q_1}{p_2-p_1}$.
\end{remark}

\section{Deficiency-one algorithm for a class of power-law kinetic systems} \label{appendix:DOA}

The proofs of Propositions \ref{prop:DOA BECCS} and \ref{prop:DOA AR} rely on the Deficiency-one algorithm (DOA). This converts the problem of determining a system's capacity to admit multiple steady states -- a non-linear problem -- into a problem involving a linear system of equations and inequalities. Feinberg \cite{FEIN1995DOA} first introduced this method to deal with regular deficiency-one mass action systems. In \cite{FMRL2018}, the algorithm was extended to a class of power law kinetic systems with deficiency one. 

\subsection{Proof of Proposition \ref{prop:DOA BECCS}}
\label{appendix:BECCS}
The CRN representation of BECCS is a regular, non-weakly reversible network with a deficiency of one. Furthermore, the CRN is discordant. Theorem 10.5.10 of Feinberg's book \cite{FEIN2019} suggests that for this network there is a weakly monotonic kinetics (such as those with non-negative kinetic orders) for which the resulting kinetic system is not injective and for which there are two distinct stoichiometrically-compatible positive steady states. 

For the necessary preliminary concepts related to the algorithm, the reader may refer to \cite{FMRL2018}.

\begin{proof}
\label{BECCS DOA}
We wish to check if there exists a $\mu=[\mu_1, \mu_2, \mu_3, \mu_4, \mu_5]^\top$ which is a solution to the linear system induced by a confluence vector and a $\{U,M,L \}$ partition of the reactant complexes, and is sign compatible with $S$.

\noindent \textit{Step 1:} Find a confluence vector $h$. A confluence vector that satisfies the three conditions in Definition 2.3 of Fortun et al. \cite{FMRL2018} is
$$ h= \begin{bmatrix} h_{A_1+2A_2} \\h_{2A_1+A_2} \\h_{A_1} \\h_{A_2} \\h_{A_3} \\h_{A_4} \\h_{A_8} \end{bmatrix} = \begin{bmatrix} -1 \\ 1 \\ -1 \\ 1 \\ 0 \\ 0 \\ 0 \end{bmatrix}. $$

\noindent \textit{Step 2:} Choose an upper-middle-lower or $\{ U,M,L \}$  partition of the reactant complexes. Consider the partition 
\begin{align*}
U &=\{ A_1+2A_2, 2A_1+A_2 \} \cup \{A_2,A_3 \} \\
M &= \{ A_1, A_4, A_8 \} \\
L &= \emptyset
\end{align*}
\noindent \textit{Step 3:} Collect all equations comparing the $T_{*,i}\cdot \mu$ among the complexes $i\in M$. Since 
\begin{equation*}
T= \begin{blockarray}{cccccccl}
A_1 + 2A_2 & 2A_1 + A_2 & A_1 & A_2 & A_3 & A_4 & A_8    \\
\begin{block}{[ccccccc]l}
 p_1 & p_2 & 1 & 0 & 0 & 0 & 0 & A_1 \\
 q_1 & q_2 & 0 & 1 & 0 & 0 & 0 & A_2 \\
 0 & 0 & 0 & 0 &  1 & 0 & 0 &  A_3  \\
0 & 0 & 0 & 0 &  0 & 1 & 0 &  A_4  \\
0 & 0 & 0 & 0 &  0 & 0 & 1 &  A_8  \\
 \end{block}
\end{blockarray}, 
\end{equation*}
we have
\begin{equation*} 
\left.
\begin{array}{rllrl}
T_{*, A_1} \cdot \mu &= T_{*, A_4}\cdot\mu & \Rightarrow& \mu_1=\mu_4, \\
T_{*, A_1} \cdot \mu &= T_{*, A_8}\cdot\mu & \Rightarrow& \mu_1=\mu_5, \\
T_{*, A_4} \cdot \mu &= T_{*, A_8}\cdot\mu & \Rightarrow& \mu_4=\mu_5. \\
  \end{array}
 \right.
\end{equation*}

\noindent \textit{Step 4:} Collect all inequalities comparing $T_{*,i}\cdot \mu$'s among the complexes in different parts. That is, if a complex $i$ lies above a complex $j$, then $T_{*,i}>T_{*,j}$. We obtain
\begin{equation*} 
\left.
  \begin{array}{rllrl}
T_{*,A_1+2A_2} \cdot \mu &> T_{*,A_1}\cdot\mu & \Rightarrow& p_1 \mu_1 + q_1\mu_2 &> \mu_1, \\
T_{*,A_1+2A_2} \cdot \mu &> T_{*,A_4}\cdot\mu & \Rightarrow& p_1 \mu_1 + q_1\mu_2 &> \mu_4, \\
T_{*,A_1+2A_2} \cdot \mu &> T_{*,A_8}\cdot\mu & \Rightarrow& p_1 \mu_1 + q_1\mu_2 &> \mu_8, \\
T_{*,2A_1+A_2} \cdot \mu &> T_{*,A_1}\cdot\mu & \Rightarrow& p_2 \mu_1 + q_2\mu_2 &> \mu_1, \\
T_{*,2A_1+A_2} \cdot \mu &> T_{*,A_4}\cdot\mu & \Rightarrow& p_2 \mu_1 + q_2\mu_2 &> \mu_4, \\
T_{*,2A_1+A_2} \cdot \mu &> T_{*,A_8}\cdot\mu & \Rightarrow& p_2 \mu_1 + q_2\mu_2 &> \mu_8, \\
T_{*,A_2} \cdot \mu &> T_{*,A_1}\cdot\mu & \Rightarrow& \mu_2 &> \mu_1, \\
T_{*,A_2} \cdot \mu &> T_{*,A_4}\cdot\mu & \Rightarrow& \mu_2 &> \mu_4, \\
T_{*,A_2} \cdot \mu &> T_{*,A_8}\cdot\mu & \Rightarrow& \mu_2 &> \mu_5, \\
T_{*,A_3} \cdot \mu &> T_{*,A_1}\cdot\mu & \Rightarrow& \mu_3 &> \mu_1, \\
T_{*,A_3} \cdot \mu &> T_{*,A_4}\cdot\mu & \Rightarrow& \mu_3 &> \mu_4, \\
T_{*,A_3} \cdot \mu &> T_{*,A_8}\cdot\mu & \Rightarrow& \mu_3 &> \mu_5. \\
  \end{array}
 \right.
\end{equation*}

\noindent \textit{Step 5:} Collect all equalities and inequalities comparing $T_{*,i}\cdot \mu$'s between complexes of a cut pair in $U$ and $L$. The complexes $A_1+2A_2$ and $2A_1+A_2$ form a cut-pair. Additionally, 
\begin{align*}
h \left( \mathscr{W}(A_1+2A_2) \right) &=\sum_{i \in \mathscr{W}(A_1+2A_2)} h_i =h_{A_1+2A_2}=-1, \\
h \left( \mathscr{W}(2A_1+A_2) \right) &=\sum_{i \in \mathscr{W}(2A_1+A_2)} h_i =h_{2A_1+A_2}=1.
\end{align*}
Hence, $$ T_{*,2A_1+A_2}\cdot \mu > T_{*,A_1+2A_2}\cdot \mu \Rightarrow p_2\mu_1+q_2\mu_2 > p_1\mu_1+q_1\mu_2.$$
$A_2 $ and $A_3 $ also form a cut-pair. Note that 
\begin{align*}
h \left( \mathscr{W}(A_2) \right) &=\sum_{i \in \mathscr{W}(A_2)} h_i =h_{A_1}+ h_{A_2}+ h_{A_4}+ h_{A_5}=0, \\
h \left( \mathscr{W}(A_3) \right) &=\sum_{i \in \mathscr{W}(A_3)} h_i =h_{A_3}=0.
\end{align*}
It follows that $T_{*,A_2}\cdot\mu = T_{*,A_3}\cdot\mu \Rightarrow \mu_2=\mu_3$.

\noindent \textit{Step 6:} Gather all relations in Step 3-5 to form a linear system. If $x=\mu_1=\mu_4=\mu_5$ and $y=\mu_2=\mu_3$, we have $$\mu=\begin{bmatrix} \mu_1 \\ \mu_2 \\ \mu_3 \\ \mu_4 \\ \mu_5 \end{bmatrix}=\begin{bmatrix}   x\\ y \\ y \\ x \\ x \end{bmatrix} \text{ with } y>x.$$
Moreover, $$x <p_1x+q_1y <p_2x+q_2y.$$
\begin{align}
\Rightarrow y &> \dfrac{(1-p_1)x}{q_1} \; \text{ if } q_1>0 , \label{ineq1} \\
y &> \dfrac{(1-p_2)x}{q_2} \; \text{ if } q_2>0, \label{ineq2}\\
y &< \dfrac{(p_2-p_1)x}{q_1-q_2} \; \text{ if } q_1>q_2.\label{ineq3} 
\end{align}

\noindent \textit{Step 7:} Check if the linear system has a nonzero solution $\mu \in \mathbb{R}^\mathscr{S}$ which is sign compatible with $S$. If such $\mu$ exists, we conclude that the CRN can admit multiple equilibria and exit the algorithm. 
$$ S =\left\lbrace  \left. a_1 \begin{bmatrix} 1 \\ -1 \\ 0 \\ 0 \\ 0 \end{bmatrix} + a_2 \begin{bmatrix} 0 \\ 1 \\ -1 \\ 0 \\ 0 \end{bmatrix} + a_3 \begin{bmatrix} 1 \\ 0 \\ 0 \\ 0 \\ -1 \end{bmatrix}+ a_4 \begin{bmatrix} 0 \\ 0 \\ 0 \\ 1 \\ -1 \end{bmatrix} \right| a_1, a_2, a_3, a_4 \in \mathbb{R} \right\rbrace. $$
Suppose $v= \begin{bmatrix} a_1+a_3 \\ -a_1+a_2 \\ -a_2 \\ a_4 \\ -a_3-a_4 \end{bmatrix} \in S$ such that $\text{sign }(v)= \begin{bmatrix} - \\ + \\ + \\ - \\ - \end{bmatrix}$. This means that $a_1<a_2<0$ and $a_1<-a_3<a_4<0$. We want to find a vector $\mu$ such that $\text{sign }(\mu)= \text{sign }(v)$. 
Note that $\mu=[\mu_1, \mu_2, \mu_3, \mu_4, \mu_5]^\top=[x, y, y, x, x]^\top $. Take $x<0$. We get $y<0$ if $p_1 >1$ in (\ref{ineq1}) and $p_2 >1$ in (\ref{ineq2}). To satisfy (\ref{ineq3}), we set $p_1 >p_2$. 
\end{proof}
In summary, we can find $\mu$ that is sign compatible with $S$ when $p_1>p_2>1$ and $q_1>q_2>0$. Hence, the system can admit multiple steady states.

\subsection{Proof of Proposition \ref{prop:DOA AR}}
\label{appendix:AR}

\begin{proof}
    Consider the reaction network of the AR system with four linkage classes and four terminal strong linkage classes.

    \begin{align*}
        \mathscr{L}_1: & {A_1+2A_2\leftrightarrows 2A_1+A_2} & \mathscr{L}_2: & {A_2\leftrightarrows A_3} \\
        \mathscr{L}_3: & {A_4+A_{15}\to A_2+A_{15}} & \mathscr{L}_4: & {A_1\to A_{15}\to A_4}
    \end{align*}

    The associated $T$ matrix of the AR system is
    \begin{equation*}
T= \begin{blockarray}{cccccccl}
A_1+2A_1 & 2A_1+A_2 & A_1 & A_2 & A_3 & A_4+A_{15} & A_{15} &  \\
\begin{block}{[ccccccc]l}
p_1 & p_2 & 1 & 0 & 0 & 0 & 0 & A_1 \\
q_1 & q_2 & 0 & 1 & 0 & 0 & 0 & A_2 \\
0 & 0 & 0 & 0 & 1 & 0 & 0 & A_3 \\
0 & 0 & 0 & 0 & 0 & e_{15} & 0 & A_4 \\
0 & 0 & 0 & 0 & 0 & f_{15} & 1 & A_{15} \\
\end{block}
\end{blockarray}
\end{equation*}

The confluence vector of an AR system is $h=\begin{bmatrix}h_{A_1+2A_2}\\h_{2A_1+A_2}\\h_{A_1}\\h_{A_2}\\h_{A_3}\\h_{A_4}\\h_{A_{15}}\\h_{A_4+A_{15}}\\h_{A_2+A_{15}}\end{bmatrix}=\begin{bmatrix}
    h_1\\h_2\\h_3\\h_4\\h_5\\h_6\\h_7\\h_8\\h_9\end{bmatrix}\in \mathbb{R}^{\mathscr{C}}$

    such that the following are satisfied:
    \begin{enumerate}
        \item $ \sum_{i\in \mathscr{C}}h_i\cdot i=h_1\begin{bmatrix}
        1\\2\\0\\0\\0\end{bmatrix}+h_2\begin{bmatrix}
        2\\1\\0\\0\\0\end{bmatrix}+h_3\begin{bmatrix}
        1\\0\\0\\0\\0\end{bmatrix}+h_4\begin{bmatrix}
        0\\1\\0\\0\\0\end{bmatrix}+h_5\begin{bmatrix}
        0\\0\\1\\0\\0\end{bmatrix}+h_6\begin{bmatrix}
        0\\0\\0\\1\\1\end{bmatrix}+h_7\begin{bmatrix}
        0\\0\\0\\0\\1\end{bmatrix}+h_8\begin{bmatrix}
        0\\0\\0\\1\\1\end{bmatrix}h_7\begin{bmatrix}
        0\\1\\0\\0\\1\end{bmatrix}=\begin{bmatrix}
        0\\0\\0\\0\\0\end{bmatrix}$
        \item For each linkage classes $\mathscr{L}_j$, $\sum_{i\in \mathscr{L}_j} h_i=0$. Hence,
        \begin{align*}
            \sum_{i\in\mathscr{L}_1} h_i= & h_1+h_2=0, & \sum_{i\in\mathscr{L}_2} h_i= & h_4+h_5=0, \\
            \sum_{i\in\mathscr{L}_3} h_i= & h_8+h_9=0, & \sum_{i\in\mathscr{L}_4} h_i= & h_3+h_6+h_7=0.
        \end{align*}
        \item For each terminal strong linkage class which is not a linkage class $\mathscr{A}$, $\sum_{i\in \mathscr{A}} h_i>0$. Since $A_4$ and $A_2+A_{15}$ are the only terminal strong non-linkage class, then 
        \[\sum_{i\in \mathscr{A}} h_i=h_6+h_9>0.\]
        
    \end{enumerate}
    
    Solving this system of linear equation, we obtain a linearly dependent vector of $h=[-1,1,-1,0,0,1,0,-1,1]^{\top}$. Let $\mu=[\mu_1,\mu_2,\mu_3,\mu_4,\mu_5]^{\top}
    \in \mathbb{R}^\mathscr{S}$. Now we choose the Upper-Middle-Lower (UML) partition of the complexes as follows:
    \begin{align*}
    U & =\{A_1+2A_2,2A_1+A_2,A_2,A_3\} & M & =\{A_1,A_4,A_{15},A_2+A_{15},A_4+A_{15}\} & L & =\emptyset
    \end{align*}
    Thus, the equalities of $T_{\cdot i}\cdot \mu$'s in $M$ are
    \begin{align*}
        T_{\cdot A_1}\cdot \mu=T_{\cdot A_{15}}\cdot \mu & \Longrightarrow \mu_1=\mu_5,\\
        T_{\cdot A_1}\cdot \mu=T_{\cdot A_4+A_{15}}\cdot \mu & \Longrightarrow \mu_1=e_{15}\mu_4+f_{15}\mu_5 \\
        T_{\cdot A_{15}}\cdot \mu=T_{\cdot A_4+A_{15}}\cdot \mu & \Longrightarrow \mu_5=e_{15}\mu_4+f_{15}\mu_5.
    \end{align*}
Since $A_2,2A_1+A_2\in U,A_1\in M$, then 
\begin{align*}
    T_{\cdot A_2}\cdot \mu>T_{\cdot A_1}\cdot \mu & \Longrightarrow \mu_2>\mu_1, \\
    T_{\cdot 2A_1+A_2}\cdot \mu>T_{\cdot A_1}\cdot \mu & \Longrightarrow p_2\mu_1+q_2\mu_2>\mu_1.
\end{align*}

Now, $A_1+2A_2$ and $2A_1+A_2$ forms a cut pair in $U$. Thus,
\[h(\mathscr{W}(A_1+2A_2))=h_1=-1,\]
\[h(\mathscr{W}(2A_1+A_2))=h_2=1.\]
Thus,
\[T_{\cdot 2A_1+A_2}\cdot \mu>T_{\cdot A_1+2A_2}\cdot \mu\Longrightarrow p_2\mu_1+q_2\mu_2>p_1\mu_1+q_1\mu_2.\]
Also, 
$A_2$ and $A_3$ forms a cut pair in $U$. Thus
\[h(\mathscr{W}(A_2))=h_4=0,\]
\[h(\mathscr{W}(A_3))=h_5=0.\]
Therefore,
\[T_{\cdot A_2}\cdot \mu=T_{\cdot A_3}\cdot \mu\Longrightarrow \mu_2=\mu_3.\]

From the previous equalities, $\mu_2=\mu_3>\mu_1=\mu_5$, but 
\[e_{4,15}\mu_4+f_{4,15}\mu_5=\mu_5\Longrightarrow \mu_5(1-f_{4,15})=e_{4,15}\mu_4\]
since $e_{4,15}+f_{4,15}=1$, then $\mu_5=\mu_4$.
Hence, $\mu_2=\mu_3>\mu_1=\mu_5=\mu_4$.
Therefore, $\mu=\begin{bmatrix}
    \mu_1\\\mu_2\\\mu_3\\\mu_4\\\mu_5\end{bmatrix}=\begin{bmatrix}
        x\\y\\y\\x\\x
    \end{bmatrix}\in \mathbb{R}^\mathscr{S}$ where $x<y$.
The rest of the computation is the same with the last part of Proof of Proposition in \ref{BECCS DOA} and will have the same result that multistationarity is achieved only when $p_1>p_2>0$ and $q_1>q_2>0$.

\end{proof}

In summary, AR system will only admit multistationarity if it is positive with $p_1>p_2>0$ and $q_1>q_2>0$.




\end{appendix}

\singlespacing


\end{document}